\documentclass[a4paper,preprint,11pt]{elsarticle}
\usepackage{amsmath,amssymb,amsthm,bm,bbm,color,float,graphicx,hyperref,multirow,stackrel}
\usepackage{breakurl}
\newtheorem{theorem}{Theorem}
\usepackage[caption=false]{subfig}
\usepackage[algosection,linesnumbered,noline,resetcount,ruled]{algorithm2e}
\SetAlgoCaptionSeparator{.}

\usepackage{pifont} % http://ctan.org/pkg/pifont
\newcommand{\cmark}{\ding{51}}
\newcommand{\xmark}{\ding{55}}
\allowdisplaybreaks % to break subequations etc.

\begin{document}
\begin{frontmatter}
\title{Full and fast calibration of \linebreak the Heston stochastic volatility model}

\author[add1]{Yiran Cui\corref{cor1}}
\ead{y.cui.12@ucl.ac.uk}
\author[add2]{Sebastian del Ba\~{n}o Rollin}
\ead{s.delbanorollin@qmul.ac.uk}
\author[add1,add3]{Guido Germano}
\ead{g.germano@ucl.ac.uk, g.germano@lse.ac.uk}

\address[add1]{Financial Computing and Analytics Group, Department of Computer Science, \linebreak University College London, United Kingdom}
\address[add2]{School of Mathematical Science, Queen Mary University of London, United Kingdom}
\address[add3]{Systemic Risk Centre, London School of Economics and Political Science, \linebreak United Kingdom}

\cortext[cor1]{Corresponding author. Tel.: +44 20 3108 7105.}

\begin{abstract}
This paper presents an algorithm for a complete and efficient calibration of the Heston stochastic volatility model. 
We express the calibration as a nonlinear least squares problem. 
We exploit a suitable representation of the Heston characteristic function and modify it to avoid discontinuities caused by branch switchings of complex functions.
Using this representation, we obtain the analytical gradient of the price of a vanilla option with respect to the model parameters, which is the key element of all variants of the objective function. The interdependency between the components of the gradient enables an efficient implementation which is around ten times faster than a numerical gradient.
We choose the Levenberg-Marquardt method to calibrate the model and do not observe multiple local minima reported in previous research.
Two-dimensional sections show that the objective function is shaped as a narrow valley with a flat bottom. Our method is the fastest calibration of the Heston model developed so far and meets the speed requirement of practical trading.
\end{abstract}

\end{frontmatter}

\section{Introduction}
Pricing financial derivatives is an established problem in the operational research literature; see for example Fusai et al.\ \cite{fusai2016} and references therein contained. Here we deal with the calibration of the Heston stochastic volatility model, which is important and popular for derivatives pricing. The particular topic of model calibration also involves numerical optimisation, which is a core subject of operational research.

% \textit{Explain the particular link of this paper to opearational research because it uses numerical optimisation, which is one of the core aspects of operational research.}

A sophisticated model may reflect the reality better than a simple one,
but usually is more challenging to implement and calibrate.
This is especially true with mathematical models for the pricing of derivatives and the estimation of risk.
The most basic model, introduced by Black and Scholes (BS) \cite{BlackScholes}, assumes that the underlying price follows a geometric Brownian motion with constant drift and volatility.
The price of a vanilla option is then given as a function of a single parameter, the volatility.
However, the BS model does not adequately take into account essential characteristics of market dynamics such as fat tails, skewness and the correlation between the value of the underlying and its volatility.
It has also been observed that the volatility starts to fluctuate when the market reacts to new information. Thus, several extensions of the BS model were suggested thereafter, including the family of stochastic volatility (SV)
models, which introduces a second Brownian motion to describe the fluctuation of the volatility.
We study one of the most important SV models; it was proposed by Heston \cite{Heston} and is defined by the system of stochastic differential equations
\begin{subequations}\label{eq:HestonSVProcess}
\begin{align}
\mathrm{d}S_t &= \mu S_t \mathrm{d}t + \sqrt{v_t}S_t \mathrm{d}W_t^{(1)},\\
\mathrm{d}v_t &= \kappa(\bar{v}-v_t) \mathrm{d}t + \sigma \sqrt{v_t}\mathrm{d}W_t^{(2)},
\end{align}
with
\begin{equation}
\mathrm{d}W_t^{(1)}\mathrm{d}W_t^{(2)} = \rho \mathrm{d}t,
\end{equation}
\end{subequations}
where $S_t$ is the underlying price and $v_t$ its variance; the parameters $\kappa, \bar{v}, \sigma, \rho$ are respectively called the mean-reversion rate, the long-term variance, the volatility of volatility, and the correlation between the Brownian motions $W_t^{(1)}$ and $W_t^{(2)}$ that drive the underlying and its variance; moreover there is a fifth parameter $v_0$, the initial value of the variance.

Model calibration is as crucial as the model itself.
Calibration consists in determining the parameter values so that the model reproduces market prices as accurately as possible.
Both the accuracy and the speed of calibration are important because practitioners use the calibrated parameters to price a large number of complicated derivative contracts and to develop high-frequency trading strategies.

In this paper, we propose to efficiently calibrate the Heston model using an analytical gradient and numerical optimisation.
In Section~\ref{sec:previouswork}, we briefly review the existing research.
In Section~\ref{sec:problem}, we formulate the objective function and derive its analytical gradient.
In Section~\ref{sec:calibration}, we present a complete algorithm to calibrate the Heston model using the Levenberg-Marquardt (LM) method.
We also discuss some points where a carefully designed numerical scheme may improve the performance.
In Section~\ref{sec:results}, we present numerical results.

Throughout, we use bold uppercase letters for matrices, e.g.\ $\bm{J}$, and bold lowercase letters for column vectors, e.g.\ $\bm{\theta}$;
 a superscript $^\intercal$ for the transpose of a matrix or vector;
$\bm{e}$ for a vector of all ones $[1, \dots, 1]^\intercal$;
 $\mathbbm{E}[\cdot]$ for expectations;
$\mathbbm{1}_A(\cdot)$ for the indicator function of the set $A$;
 $\mathrm{Re}(\cdot)$ for the real part and $\mathrm{Im}(\cdot)$ for the imaginary part of a complex number;
$\|\cdot\|$ for the $l_2$-norm;
$\|\cdot\|_\infty$ for the $l_\infty$-norm;
$\log$ for the natural logarithm.

\section{Previous work on Heston model calibration}\label{sec:previouswork}
In the literature, there are two main approaches to calibrate the Heston model: \textit{historical} and \textit{implied}.
%are two main approaches to calibrate the Heston model.
The first fits historical time series of the prices of an option with a fixed strike and maturity, typically by the maximum likelihood method or the efficient method of moments \cite{ai2007maximum,fatone2014calibration,hurn2014estimating}. %, seeking for the
%parameter set that reproduces the historical prices as accurately as possible.
The second fits the volatility surface of an underlying at a fixed time, i.e., options with several strikes and
maturities, to obtain the implied parameter set.
Our work follows the second approach, as that is what is used in real-time pricing and risk management.
%which requires a fast speed of calibration for the real-time use.
In the following, we survey obstacles and existing methods for
the Heston model calibration related to the second approach.

\subsection{Recognised difficulties}
Firstly, the calibration is in a five-dimensional space. There is no consensus among researchers on
whether the objective function for the Heston model calibration is convex or irregular.
The results of some proposed methods \cite{ChenBin, GilliCalibrating, MikhailovNogel} depend on the initial point, which
was attributed to a non-convex objective function, but might simply reflect on the inadequacy of the methods.
%Researchers  found different optima starting from different initial points and attribute this to a non-convex objective function .
To find a reasonable initial guess, short-term or long-term asymptotic rules are used;
see Jacquier and Martini \cite{ZeliadeHeston} for a detailed review.
However, recently Gerlich et al.~\cite{gerlich2012parameter}
claimed a convergence to the unique solution independent of the initial guess and suggested
that the Heston calibration problem may have some inherent structure leading to a single stationary point.
On the other hand, dependencies among the parameters do exist. For example, it is known that
$\sigma$ and $\kappa$ offset each other: it is possible that a parameter set with
certain values of $\sigma$ and $\kappa$ gives a fit comparable to a set with different values of $\sigma$ and $\kappa$.
Intuitively, the fact that different parameter combinations yield similar values of the objective function could be due to the objective function being flat close to the optimum; see Section~\ref{sec:results} in this paper.

Secondly, the analytical gradient for the Heston calibration problem is hard to find and has not been available so far because it was believed that the expression of the Heston characteristic function is overly complicated to provide an insightful analytical gradient: of course, a gradient can be obtained with symbolic algebra packages, but the resulting expressions are intractable.
Instead, numerical gradients obtained by finite difference methods have been used in gradient-based optimisation methods;
%, which typically search through a sequence of iterates by moving along a descent direction that forms an acute angle with the gradient,
however, numerical gradients have a larger computational cost and a lower accuracy.
%In the current literature there is no analytical formula of the gradient for the Heston calibration problem

\subsection{Existing methods}
We review some heuristics to reduce the dimension of the calibration and then the optimisation methods
that have been applied so far.

\subsubsection{Heuristics for dimension reduction}
Since the Heston parameters are closely related to the shape of the implied volatility surface \cite{Clark, Gatheral2006volatility, GilliCalibrating,janek2011fx} ($v_0$  controls the position of the volatility smile, $\rho$ the skewness, $\kappa$ and $\sigma$ the convexity, and $\kappa$ times the difference between $v_0$ and $\bar{v}$ the term structure of implied volatility), efforts have been made to simplify the calibration to a lower dimension by presuming some of the parameter values 
% based on the knowledge of the volatility surface at hand. 
based on knowledge available for the specific volatility surface.
The initial variance $v_0$ is usually set to the short-term at-the-money (ATM) BS implied variance, which is based on the term structure of the
BS implied volatility in the Heston model \cite[p.~34-35]{Gatheral2006volatility}. A practical
calibration experiment \cite[p.~29-30]{ChenBin} verified the linearity between the initial variance and the BS
implied variance for maturities in the range of 1 to 2 months. Clark \cite[Eq.~(7.3)]{Clark} suggested the
heuristic assumption $\kappa = 2.75/T$ and $\bar{v} = \sigma_{\mathrm{ATM}}(T)$, where $\sigma_{\mathrm{ATM}}(T)$
is the ATM BS implied volatility with maturity $T$. Chen \cite{ChenBin} proposed a fast intraday recalibration
by fixing $\kappa$ to the same as yesterday's and $v_0$ to the 2-month ATM implied volatility, which are
heuristics actually adopted in the industry. These assumptions help with an incomplete calibration, but
may misguide the iterate to a limited domain and thus to a wrong convergence point.

\subsubsection{Stochastic optimisation methods}
Researchers who believed that a descent
direction is unavailable have devoted their attention to
stochastic optimisation methods, including
Wang-Landau \cite{ChenBin}, differential evolution and particle swarm \cite{GilliHeuristic},
%(adapted)
simulated annealing \cite{moodley2009heston}, etc. To increase the robustness,
a deterministic search such as Nelder and Mead using the MATLAB function {\tt{fminsearch}}
is often combined with these stochastic optimisation algorithms. Almost all research using stochastic techniques
reports issues with performance. GPU technology has been applied with simulated
annealing to speed up the calibration \cite{GPU}, however a speed of $9.7$ hours with one GPU is still too
slow for real-time use.

\subsubsection{Deterministic optimisation methods}
Deterministic optimisation solvers available with commercial packages have been proved to be
unstable as the performance largely depends on the quality of the initial guess: this applies to the Excel built-in solver \cite{MikhailovNogel} and to the MATLAB solver
{\tt{lsqnonlin}} \cite{BauerFast2012, fusai2008, moodley2009heston}. Gerlich et al.~\cite{gerlich2012parameter}
adopted a Gauss-Newton framework and kept the feasibility of the iterates by projecting to
a cone determined by the constraints. The gradient of the objective function was calculated by
finite differences and thus costs a large number of function evaluations.
\medskip

To sum up, existing calibration algorithms are either based on ad hoc assumptions or not fast or stable enough for practical use. In this work, we will focus on deterministic optimisation methods without any presumption on the values of the parameters.

\section{Problem formulation and gradient calculation}\label{sec:problem}
The idea of calibrating a volatility model is to minimise the difference between the vanilla option price calculated with the model and the one observed in the market. In this section, we first formulate the calibration problem in a least squares form. Then, we present the pricing formula of a vanilla option under the Heston model with four algebraically equivalent representations of the characteristic function, discussing their numerical stability and suitability for analytical derivation. We calculate the analytical gradient of the objective function which can be used in any gradient-based optimisation algorithm.

\subsection{The inverse problem formulation}\label{sec:TheInverseProblem}
%The popularity of the Heston model comes from its analytical formula \eqref{eq:HestonCall} for pricing a vanilla call option.
Denote by $C^*(K_i,T_i)$ the market price of a vanilla call option with strike $K_i$ and
maturity $T_i$, $C(\bm{\theta}; K_i, T_i)$ the price computed via the Heston analytical formula \eqref{eq:HestonCall} with the parameter vector $
\bm{\theta}:=\left[v_0, \bar{v}, \rho,\kappa,\sigma \right]^\intercal$. % of dimension $m=5$.
We assemble the residuals for the $n$ options to be calibrated
\begin{equation}\label{eq:def_rj}
    r_i(\bm{\theta}):= C(\bm{\theta}; K_i, T_i) - C^*(K_i, T_i), \qquad i = 1, \dots, n
\end{equation}
in the residual vector $\bm{r}(\bm{\theta}) \in \mathbb{R}^{n}$, i.e.,
\begin{equation}\label{eq:def_r}
    \bm{r}(\bm{\theta}):= \left[r_1(\bm{\theta}),r_2(\bm{\theta}), \dots ,r_n(\bm{\theta})\right]^\intercal.
\end{equation}

We treat the calibration of the Heston model as an inverse problem in the nonlinear least squares form
\begin{equation}\label{eq:LSform}
    \min_{\bm{\theta}\in\mathbb{R}^{m}} f(\bm{\theta}),
\end{equation}
where $m=5$ indicates the dimension, and
\begin{equation}\label{eq:defOfLS}
    f(\bm{\theta}):=  \frac{1}{2}\|\bm{r}(\bm{\theta})\|^2=\frac{1}{2} \bm{r}^\intercal(\bm{\theta}) \bm{r}(\bm{\theta}).
    %= \frac{1}{2}\bm{r}(\bm{\theta})\cdot\bm{r}(\bm{\theta})
%    = \frac{1}{2}\|\bm{r}(\bm{\theta})\|^2
   % = \frac{1}{2} \sum_{i=1}^{n} r_i^2(\bm{\theta}).
\end{equation}
%In \eqref{eq:defOfLS} we list four equivalent expressions for the sum of the squares of residuals;
%In the following we will use the first and third.
Since  there are many more market observations than parameters to be found, i.e., $n\gg m$, the calibration problem is overdetermined.

Before applying any technique to solve the problem \eqref{eq:LSform}--\eqref{eq:defOfLS}, one needs to bear in mind that the evaluation of $C(\bm{\theta}; K_i, T_i)$ is expensive for the purpose of
calibration; hence, one would like to minimise the number of computations of Eq.~\eqref{eq:HestonCall} when designing the algorithm.
Moreover, the explicit gradient of $C(\bm{\theta}; K_i, T_i)$ with respect to $\bm{\theta}$ is not available in the literature as it is deemed to be overly complicated.
This is indeed true if one starts from the commonly used expressions for the characteristic function by Heston \cite[Eq.~(17)]{Heston} or Schoutens et al.~\cite[Eq.~(17)]{schoutens2003perfect}.
However, as shown in the next section, a more convenient choice of the functional form of the characteristic function by del Ba\~{n}o Rollin et al.~\cite[Eq.~(6)]{delBano2010}
eases the derivation of its analytical gradient.

\subsection{Pricing formula of a vanilla option and representations of the characteristic function}
For a spot price $S_0$ and an interest rate $r$, the price of a vanilla call option with strike $K$ and maturity $T$ is
\begin{align}\label{eq:callpv}
C(\bm{\theta}; K,T) &= e^{-rT} \mathbbm{E}[(S_T-K) \mathbbm{1}_{\left\{S_T\geq K\right\} }(S_T)]\\
           &=e^{-rT}\left( \mathbbm{E}[S_T \mathbbm{1}_{\left\{S_T\geq K\right\} }(S_T)] - K \mathbbm{E}[ \mathbbm{1}_{\left\{S_T\geq K\right\} }(S_T)]\right) \nonumber \\
           &= S_0P_1 - e^{-rT}KP_2. \nonumber
\end{align}
In the Heston model, $P_1$ and $P_2$ are solutions to certain pricing PDEs \cite[Eq.~(12)]{Heston} and are given as
\begin{align}
P_1 &= \frac{1}{2}+\frac{1}{\pi} \int_0^\infty \mathrm{Re} \left( \frac{e^{-iu \log K} }{iu F} \phi(\bm{\theta};u-i,T)\right) \mathrm{d}u,\\
P_2 &= \frac{1}{2}+\frac{1}{\pi} \int_0^\infty \mathrm{Re} \left( \frac{e^{-iu \log K} }{iu}\phi(\bm{\theta};u,T) \right) \mathrm{d}u,
\end{align}
where $i$ is the imaginary unit, $F:=S_0e^{rT}$ is the forward price and $\phi(\bm{\theta};u,t)$ is the characteristic function of the logarithm of the stock price process. Thus, the formula for
pricing a vanilla call option becomes
\begin{multline}\label{eq:HestonCall}
    C(\bm{\theta};K,T) = \frac{1}{2} (S_0 - e^{-rT}K)
    +\frac{e^{-rT}}{\pi}  \left[ \int_0^\infty\mathrm{Re} \left(\frac{e^{-iu\log K}}{iu} \phi(\bm{\theta};u-i, T)\right)\mathrm{d} u \right.  \\ \left. - K\int_0^\infty \mathrm{Re} \left(\frac{e^{-iu\log K}
}{iu}\phi(\bm{\theta};u, T)\right) \mathrm{d} u\right].
\end{multline}
The characteristic function was originally given by Heston as \cite[Eq.~(17)]{Heston}
\begin{multline}\label{eq:HestonCharFn}
    \phi(\bm{\theta};u, t) = \exp\left\{iu(\log S_0 + rt)
    %\phantom{\frac{1-g_1e^{dt}}{1-g_1}}
    + \frac{\kappa\bar{v}}{\sigma^2} \left[ (\xi + d)t - 2\log{\frac{1-g_1e^{dt}}{1-g_1}}\right] \right. \\ + \left.\frac{v_0}{\sigma^2} (\xi + d) \frac{1-e^{dt}}{1-g_1e^{dt}}\right\},
\end{multline}
where
\begin{subequations}\label{eq:CharFn_subdefinitions}
\begin{align}
    \xi &:= \kappa-\sigma\rho i u,\label{eqeq:CharFn_subdefinitions_xi}\\
    d &:= \sqrt{\xi^2 + \sigma^2 (u^2+iu)}, \label{eq:CharFn_subdefinitions_d}\\
    g_1 &:= \frac{\xi+d}{\xi -d}.
\end{align}
\end{subequations}

Kahl and J\"ackel~\cite{Kahl2005} pointed out that when evaluating this form as a function of $u$ for moderate to long maturities, discontinuities appear because of the branch switching of the complex power function
$G^\alpha(u) = \exp(\alpha\log G(u))$ with $G(u) := (1-g_1e^{dt})/(1-g_1)$ and $\alpha := \kappa\bar{v}/\sigma^2$, which appears in Eq.~\eqref{eq:HestonCharFn} as a multivalued complex logarithm.
% logarithm in \eqref{eq:HestonCharFn} and the complex square root in \eqref{eq:CharFn_subdefinitions_d}.
This depends on the fact that $G(u)$ has a shape of a spiral as $u$ increases, and when it repeatedly crosses the negative real axis, the phase of $G(u)$ jumps from $-\pi$ to $\pi$. Then the phase of $G^\alpha(u)$ changes from $-\alpha\pi$ to $\alpha\pi$, causing a discontinuity when $\alpha$ is not a natural number.

Albrecher et al.~\cite{Albrecher2007} found that this happens when the principal value of the complex square root $d$ is selected, as most numerical implementations of these functions do, but can be avoided if the second value is used instead.
% A branch cut is a curve in the complex plane across which a complex function is discontinuous. The negative real axis is a branch cut for both the logarithm and the square root of a complex number.
% The complex square root $d$ has two possible values which are opposite to each other. Albrecher et al.~\cite{Albrecher2007} found that by selecting the principal value of $d$, as most numerical implementations of these functions do, $G$ is a spiral with exponentially growing radius which repeatedly crosses the negative real axis when $u$ increases.
% Thus, when using this as an argument of the logarithm, discontinuities appear and give rise to numerical instabilities and potential mispricings.
% They further showed that the discontinuities disappear if the second value of $d$ is used instead, which can be obtained by flipping the sign of $d$.
They proved that this alternative representation, originally proposed by Schoutens et al.~\cite[Eq.~(17)]{schoutens2003perfect}, is continuous and gives numerically stable prices in the full-dimensional and unrestricted parameter space:
\begin{multline}\label{eq:SchoutensCharFn}
    \phi(\bm{\theta}; u, t) = \exp\left\{
    iu(\log S_0 + rt)+\frac{\kappa \bar{v}} {\sigma^2} \left[ (\xi - d)t - 2\log{\frac{1-g_2e^{-dt}}{1-g_2}}\right]\right. \\ + \left.\frac{v_0}{\sigma^2} (\xi - d) \frac{1-e^{-dt}}{1-g_2e^{-dt}}\right\},
\end{multline}
where
\begin{equation}
g_2 := \frac{\xi-d}{\xi + d} = \frac{1}{g_1}.
\end{equation}

Another equivalent form of the characteristic function was proposed later by del Ba\~{n}o Rollin et al.~\cite[Eq.~(6)]{delBano2010}. We correct the expression in that paper by adding the term $-{t \kappa \bar{v} \rho i u }/{\sigma}$ to the exponent, resulting in
\begin{equation}\label{eq:SebCharFn}
    \phi(\bm{\theta};u, t) = \exp\left[ iu(\log S_0 + rt)
    - \frac{t \kappa \bar{v} \rho i u }{\sigma}
    - v_0A\right] B^{2\kappa \bar{v}/{\sigma^2}},
\end{equation}
where
\begin{subequations}\label{eq:SebCharFn_subdefinitions}
    \begin{align}
    A &:= \frac{A_1}{A_2},\\
    A_1 &:= (u^2+iu)\sinh \frac{dt}{2},\\
    A_2 &:= d\cosh \frac{dt}{2}+\xi \sinh \frac{dt}{2},\label{eq:A2}\\
    B &:=\frac{de^{\kappa t/2}} {A_2}.
   \end{align}
\end{subequations}
Del Ba\~{n}o Rollin et al.\ introduced their expression to analyse the log-spot density, and since then it has not been used for any other purpose. It was obtained by manipulating the complex moment generating function; besides being more compact, it replaces the exponential functions in the exponent with hyperbolic functions, which makes the derivatives easier. Therefore, we will use this expression to obtain the analytical gradient.

However, the same discontinuity problem pointed out by Kahl and J\"ackel appears here too. It comes from the factor $B^{2\kappa\bar{v}/\sigma^2}$, or more specifically from the denominator of $B$, i.e., $A_2$.
Fig.~\ref{fig:trajectory_A2} shows a trajectory of $\gamma(u):= (A_2(u)\log\log |A_2(u)|)/|A_2(u)|$. The double-logarithmic scaling of the radius compensates the rapid outward movement of the spiralling trajectory of $A_2(u)$ \cite{Albrecher2007,Kahl2005}. For the curve we adopt the same hue $h\in[0,1)$ as Kahl and J\"ackel \cite{Kahl2005}, $h:=\log_{10}(u+1) \mod 1$, which means that segments of slowly varying colour represent rapid movements of $A_2(u)$ as a function of $u$.
\begin{figure}[h]
   \subfloat[{$\gamma(u), \ u\in [0,500]$.}]{\includegraphics[width=0.48\textwidth]{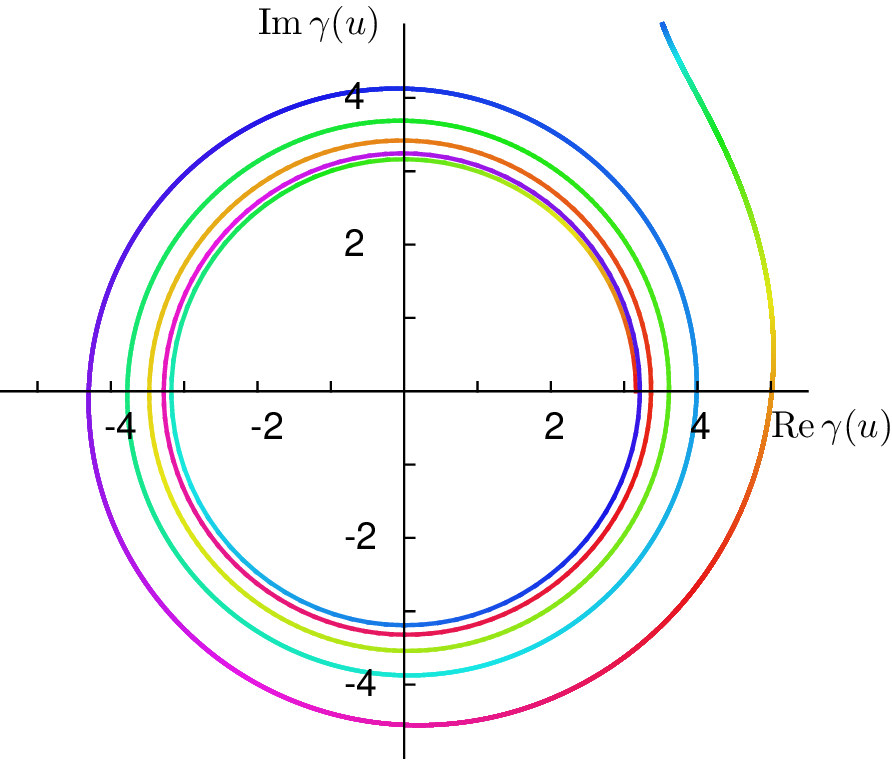}\label{fig:trajectory_A2}}\hfill
   \subfloat[{$\log A_2(u),\ u\in [0,4]$.}]{\includegraphics[width=0.48\textwidth]{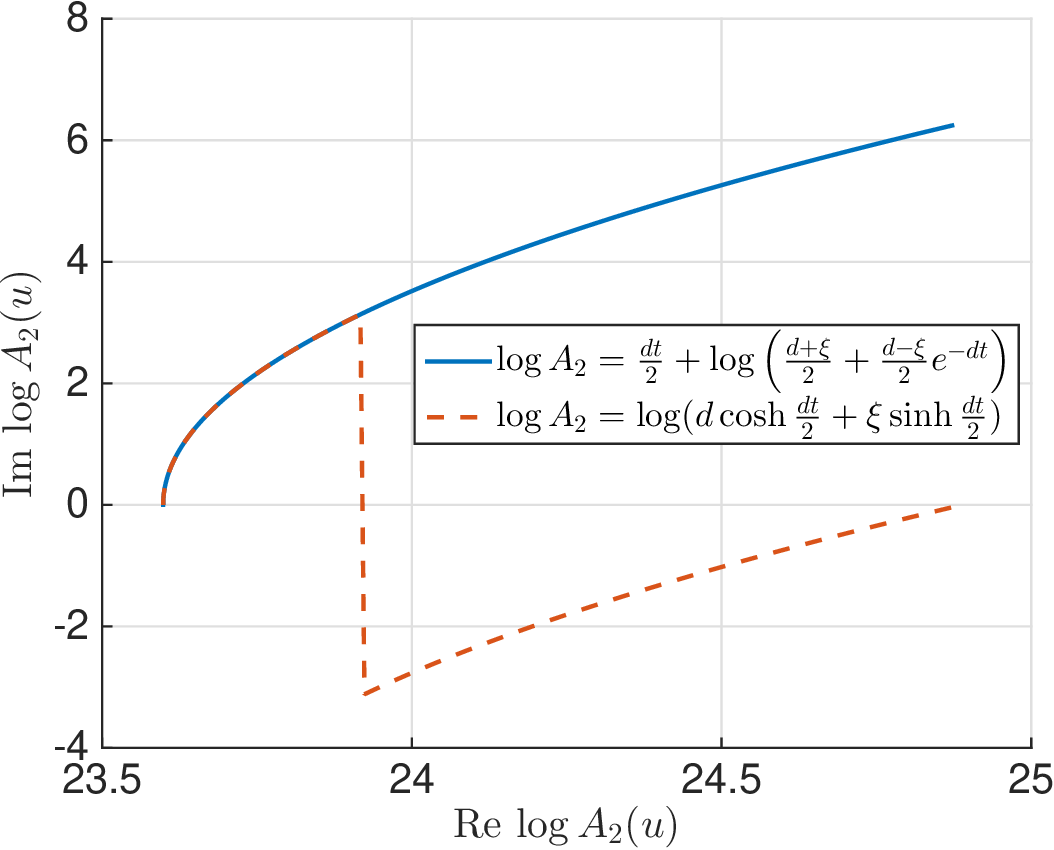}\label{fig:trajectory_logA2}}
\caption{Trajectories of $\gamma(u)$ and two equivalent forms of $\log A_2(u)$ in the complex plane. The curves were generated using the parameters in Table~\ref{tab:parameter_exp1} with maturity $T = 15$.\label{fig:trajectory}}
\end{figure}
\begin{table}[h]
\caption{Parameters specification.}
\begin{center}\footnotesize
\renewcommand{\arraystretch}{1.3}
\begin{tabular}{cr|cr}\hline\hline
\multicolumn{2}{c|}{Model parameters} & \multicolumn{2}{c}{Market parameters}\\\hline
     $\kappa$  &  $3.00$ & $S_0$ & $1.00$ \\
     $\bar{v}$ &  $0.10$ & $K$   & $1.10$ \\
     $\sigma$  &  $0.25$ & $r$   & $0.02$ \\
     $\rho$    & $-0.80$ &       &        \\
     $v_0$     &  $0.08$ &       &        \\ \hline\hline
\end{tabular}\label{tab:parameter_exp1}
\end{center}
\end{table}

We thus modify the representation by rearranging $\log A_2$ to
\begin{subequations}\label{eq:logA2}
\begin{align}
\log A_2 &= \log \left( d\cosh \frac{dt}{2} + \xi \sinh \frac{dt}{2} \right)\\
    &= \log \left( d\frac{e^{dt/2} + e^{-dt/2}}{2} + \xi \frac{e^{dt/2} - e^{-dt/2}}{2} \right)\\
    &= \log \left( \frac{d+\xi}{2}e^{dt/2} + \frac{d-\xi}{2}e^{-dt/2} \right)\\
    &= \log \left[ e^{dt/2} \left(\frac{d+\xi}{2} + \frac{d-\xi}{2}e^{-dt} \right)\right]\\
    &= \frac{dt}{2} + \log\left(\frac{d+\xi}{2} + \frac{d-\xi}{2}e^{-dt} \right).\label{eq:logA2_final}
\end{align}
\end{subequations}
Fig.~\ref{fig:trajectory_logA2} shows the trajectories of the two equivalent formulations of $\log A_2$. The rearrangement \eqref{eq:logA2_final} resolves the discontinuities arising from the logarithm with Eq.~\eqref{eq:A2} as an argument.
Then we insert Eq.~\eqref{eq:logA2_final} into $\log B$ and denote the final expression as $D$:
 % = \exp[(2\kappa \bar{v}/{\sigma^2}) \log B]$.
%, or more precisely $\log B$ which is a multivalued function for complex values of $B$.
% We modify it by modifying $\log B$ as follows. By definition,
\begin{subequations}\label{eq:logB}
\begin{align}
\log B &= \log d + \frac{\kappa t}{2} - \log A_2\\
       &= \log d + \frac{(\kappa-d) t}{2}  - \log\left(\frac{d+\xi}{2} + \frac{d-\xi}{2} e^{-dt}\right)
       =: D.\label{eq:D}
\end{align}
\end{subequations}
% \begin{equation}\label{eq:D}
% \log B = \log d + \frac{(\kappa-d) t}{2}  - \log\left(\frac{d+\xi}{2} + \frac{d-\xi}{2} e^{-dt}\right) =: D.
% \end{equation}
% and $\log A_2$ is the cause of the discontinuity.
% Fig.~\ref{fig:trajectory_A2} shows a trajectory of $A_2(u)$ in the complex plane, crossing the negative real axis as $u$ increases.

So we propose a new representation of the characteristic function which is algebraically equivalent to all the previous expressions and does not show the discontinuities of Eqs.~\eqref{eq:HestonCharFn} and \eqref{eq:SebCharFn} for large maturities:
\begin{equation}\label{eq:CuiCharFn}
    \phi(\bm{\theta};u,t) = \exp\left\{ iu(\log S_0 + rt)
    - \frac{t \kappa \bar{v} \rho i u }{\sigma} - v_0A
    + \frac{2\kappa\bar{v}}{\sigma^2} D \right\}.
\end{equation}

\begin{figure}[h]
\centering
\includegraphics[width=0.48\textwidth]{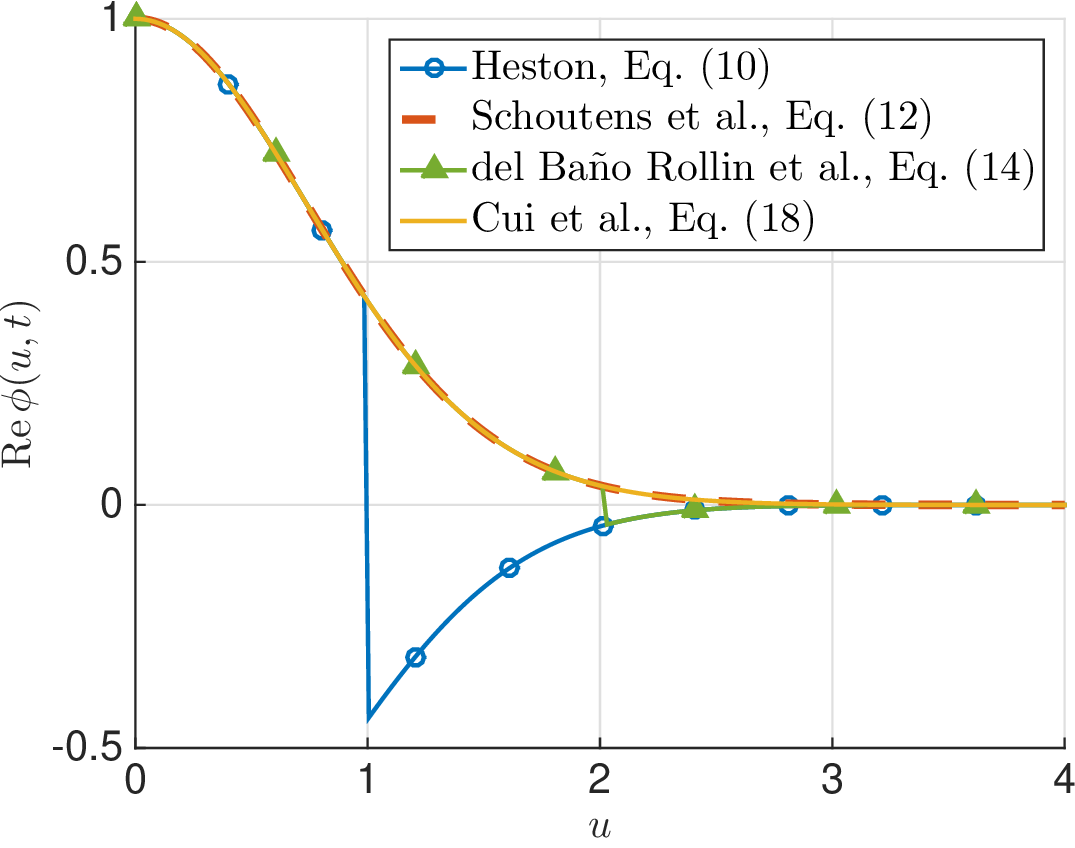}
\caption{Four equivalent representations of the Heston characteristic function. The curves were generated using the parameters in Table~\ref{tab:parameter_exp1} with maturity $T = 15$. Eq.~\eqref{eq:HestonCharFn} jumps at $u=1$, Eq.~\eqref{eq:SebCharFn} jumps at $u=2$, while Eqs.~\eqref{eq:SchoutensCharFn} and \eqref{eq:CuiCharFn} are continuous.\label{fig:charfncomparison}}
\end{figure}

We have discussed four equivalent representations of the Heston characteristic function, three from previous research and one newly proposed here by us. We compare them in Fig.~\ref{fig:charfncomparison}: the plot of our expression is continuous and overlaps Schoutens et al.'s, while the other two exhibit discontinuities due to the multivalued complex functions. Moreover our expression, like the one by del Ba\~{n}o Rollin et al.\ from which it was obtained, has the advantage of being easily derivable, as shown in the next section. These properties are summarised in Table~\ref{tab:charfn}.

\begin{table}[h]
\caption{Properties of the four representations of the Heston characteristic function.}\label{tab:charfn}
\begin{center}\footnotesize
\renewcommand{\arraystretch}{1.3}
\begin{tabular}{l|cc}\hline\hline
                & Numerically continuous & Easily derivable \\
     \hline
     Heston                 & \xmark     & \xmark  \\
     Schoutens et al.       & \cmark     & \xmark  \\
     del Ba\~{n}o Rollin et al. & \xmark     & \cmark  \\
     Cui et al.             & \cmark     & \cmark  \\ \hline\hline
\end{tabular}
\end{center}
\end{table}

\subsection{Analytical gradient}\label{sec:analyticalGradient}
We use $\bm{\nabla} = \partial/\partial \bm{\theta}$ for the gradient operator with respect to the parameter vector $\bm{\theta}$ and $\bm{\nabla}\bm{\nabla}^\intercal$ for the Hessian operator. For convenience, we omit to write the dependence of the residual vector $\bm{r}$ on $\bm{\theta}$.

\subsubsection{The basic theorem of the analytical gradient}
Let $\bm{J} =\bm{\nabla} \bm{r}^\intercal \in \mathbb{R}^{m\times n}$ be the Jacobian matrix of the residual vector $\bm{r}$ with elements
\begin{equation}
J_{ji} = \left[ \frac{\partial r_i}{\partial \theta_j} \right]= \left[ \frac{\partial C(\bm{\theta}; K_i, T_i)}{\partial \theta_j} \right],
\end{equation}
and $\bm{H}(r_i):=\bm{\nabla}\bm{\nabla}^\intercal r_i \in \mathbb{R}^{m\times m}$ be the Hessian matrix of each residual $r_i$ with elements
\begin{equation}
H_{jk}(r_i)= \left[ \frac{\partial^2 r_i}{\partial \theta_j\partial \theta_k} \right].
\end{equation}
Following the nonlinear least squares formulation \eqref{eq:LSform}--\eqref{eq:defOfLS}, one can easily write the gradient and Hessian of the objective function $f$ as
\begin{subequations}\label{eq:JacobianHessianForGeneralLS}
\begin{align}
\bm{\nabla} f &= \bm{J}\bm{r},\label{JacobianLS}\\
\bm{\nabla}\bm{\nabla}^\intercal f &= \bm{J} \bm{J}^\intercal + \sum^n_{i=1} r_i \bm{H}(r_i). \label{HessianLS}
\end{align}
\end{subequations}

\begin{theorem}
Assume that an underlying asset $S$ follows the Heston process \eqref{eq:HestonSVProcess}.
Let $\bm{\theta}:=\left[v_0, \bar{v}, \rho,\kappa,\sigma \right]^\intercal$
%\kappa,\bar{v},\rho,\sigma \right]^\intercal$
be the parameters in the Heston model, $C(\bm{\theta}; K,T)$ be the price of a vanilla call option on $S
$ with strike $K$ and maturity $T$.
Then the gradient of $C(\bm{\theta}; K,T)$ with respect to $\bm{\theta}$ is
\begin{multline}\label{eq:dC_dtheta}
\bm{\nabla} C(\bm{\theta}; K,T) = \frac{e^{-rT}}{\pi} \left[ \int_0^\infty\mathrm{Re} \left(\frac{K^{-iu}}{iu}\bm{\nabla}\phi(\bm{\theta};u-i, T)\right)
\mathrm{d} u \right. \\ - \left. K\int_0^\infty  \mathrm{Re} \left(\frac{K^{-iu}}{iu} \bm{\nabla} \phi(\bm{\theta};u, T) \right)\mathrm{d}u\right],
\end{multline}
where $\bm{\nabla} \phi(\bm{\theta}; u, T) = \phi(\bm{\theta}; u, T)\bm{h}(u)$, $\bm{h}(u):= \left[ h_1(u), h_2(u), \dots, h_5(u)\right]^\intercal$ with elements
\begin{subequations}\label{eq:dphi_dtheta}
\begin{align}
\label{eq:dphi_dv0}h_1(u) &=  -A,\\
\label{eq:dphi_dvinf}h_2(u) &= \frac{2\kappa}{\sigma^2}D - \frac{t\kappa \rho i u}{\sigma},\\
\label{eq:dphi_drho}h_3(u) &= -v_0\frac{\partial A}{\partial \rho} + \frac{2\kappa \bar{v}}{\sigma^2 d}\left(\frac{\partial d}{\partial \rho} - \frac{d}{A_2}\frac{\partial A_2}{\partial
\rho}\right) - \frac{t\kappa \bar{v} i u}{\sigma},\\
\label{eq:dphi_dkap}h_4(u) &= \frac{v_0}{\sigma i u} \frac{\partial A}{\partial \rho} + \frac{2\bar{v}}{\sigma^2}D + \frac{2\kappa \bar{v}}{\sigma^2B} \frac{\partial B}{\partial
\kappa} - \frac{t \bar{v} \rho i u}{\sigma},\\
\label{eq:dphi_dsig}h_5(u) &= -v_0 \frac{\partial A}{\partial \sigma} - \frac{4\kappa \bar{v}}{\sigma^3} D + \frac{2\kappa \bar{v}}{\sigma^2 d}\left(\frac{\partial d}{\partial
\sigma} - \frac{d}{A_2}\frac{\partial A_2}{\partial \sigma}\right) + \frac{t \kappa \bar{v} \rho i u}{\sigma^2};
\end{align}
\end{subequations}
$\xi, d, A, A_1, A_2, B, D, \phi(\bm{\theta}; u, T)$ are defined in Eqs.~\eqref{eqeq:CharFn_subdefinitions_xi}, \eqref{eq:CharFn_subdefinitions_d}, \eqref{eq:SebCharFn_subdefinitions}, \eqref{eq:D} and \eqref{eq:CuiCharFn}, respectively.
\end{theorem}

\begin{proof}
Eq.~\eqref{eq:dC_dtheta} is a direct result from the vanilla option pricing function \eqref{eq:HestonCall}.
Then the problem reduces to the derivation of the gradient of the characteristic function $\phi(\bm{\theta}; u, T)$.
Starting from Eq.~\eqref{eq:SebCharFn} and following the chain rule, one can get $\bm{\nabla} \phi(\bm{\theta}; u, T)$ as discussed below.

Since $v_0$ and $\bar{v}$ are only in the exponent and are not involved with the definition of $A$ or $B$, we directly obtain
\begin{gather}
\frac{\partial \phi(\bm{\theta}; u, T)}{\partial v_0} = -A\phi(\bm{\theta}; u, T),\\
\frac{\partial \phi(\bm{\theta}; u, T)}{\partial \bar{v}} = \frac{2\kappa \log{B}\phi(\bm{\theta}; u, T)}{\sigma^2}.\label{proof:dphi_dvbar}
\end{gather}
Next we derive the partial derivative with respect to $\rho$, since it provides some terms that can be reused for the rest.
We have
\begin{subequations}\label{eq:deriv_rho}
\begin{align}
\frac{\partial \phi(\bm{\theta}; u, T)}{\partial \rho} &= \phi(\bm{\theta}; u, T)\left(-\frac{t\kappa \bar{v}iu}{\sigma}-v_0\frac{\partial A}{\partial \rho}\right)
+\phi(\bm{\theta}; u, T)\frac{2\kappa\bar{v}}{\sigma^2}\frac{1}{B}\frac{\partial B}{\partial \rho}\\
% \begin{split}
% &= \phi(\bm{\theta}; u, T) \biggl[ -\frac{t\kappa \bar{v}iu}{\sigma}-v_0\frac{\partial A}{\partial \rho} \\%\right. \\ \left.
% &\qquad \qquad \qquad \qquad \ \ +  \frac{2\kappa\bar{v}}{\sigma^2}\frac{A_2}{de^{\kappa t/2}}\frac{e^{\kappa t/2}}{A_2}
% \left( \frac{\partial d}{\partial \rho} - \frac{d}{A_2} \frac{\partial A_2}{\partial \rho} \right)\biggr]
% \end{split}
% \\
&= \phi(\bm{\theta}; u, T) \left[ -\frac{t\kappa \bar{v}iu}{\sigma}-v_0\frac{\partial A}{\partial \rho} +  \frac{2\kappa\bar{v}}{\sigma^2d}
\left( \frac{\partial d}{\partial \rho} - \frac{d}{A_2} \frac{\partial A_2}{\partial \rho} \right)\right] \\
&= \phi(\bm{\theta}; u, T) \left[-v_0\frac{\partial A}{\partial \rho} + \frac{2\kappa \bar{v}}{\sigma^2 d}\left(\frac{\partial d}{\partial \rho} - \frac{d}{A_2}\frac{\partial A_2}{\partial
\rho}\right) - \frac{t\kappa \bar{v} i u}{\sigma}\right],
\end{align}
\end{subequations}
where
\begin{subequations}\label{eq:fragmentsRho}
\begin{align}
\frac{\partial d}{\partial \rho} &= -\frac{\xi \sigma i u}{d},\\
\frac{\partial A_2}{\partial \rho} &= -\frac{\sigma i u (2+t\xi)}{2d}\left(\xi \cosh \frac{d t}{2} + d \sinh \frac{d t}{2}\right),\\
\frac{\partial B}{\partial \rho}&= e^{\kappa t/2}\left(\frac{1}{A_2}\frac{\partial d}{\partial \rho} - \frac{d}{A_2^2}\frac{\partial A_2}{\partial \rho}\right),\\
\frac{\partial A_1}{\partial \rho} &= - \frac{i u(u^2+iu)t\xi \sigma }{2d}\cosh \frac{dt}{2} ,\\
\frac{\partial A}{\partial \rho} &= \frac{1}{A_2} \frac{\partial A_1}{\partial \rho} - \frac{A}{A_2} \frac{\partial A_2}{\partial \rho}.
\end{align}
\end{subequations}
%We obtain ${\partial \phi(\bm{\theta}; u, T)}/{\partial \rho}$ by inserting \eqref{eq:fragmentsRho} into \eqref{eq:dphi_drho}.
By merging and rearranging terms, we find that
\begin{subequations}\label{eq:fragmentsKap}
\begin{align}
\frac{\partial A}{\partial \kappa}&= \frac{i}{cu} \frac{\partial A}{\partial \rho},\\
\frac{\partial B}{\partial \kappa}&=  \frac{i}{\sigma u} \frac{\partial B}{\partial \rho} + \frac{tB}{2},
\end{align}
\end{subequations}
which are inserted into
\begin{equation}\label{proof:dphi_dkappa}
\frac{\partial \phi(\bm{\theta}; u, T)}{\partial \kappa} =  \phi(\bm{\theta}; u, T) \left( -v_0 \frac{\partial A}{\partial \kappa} +  \frac{2\bar{v}}{\sigma^2}\log B+
 \frac{2\kappa \bar{v}}{\sigma^2B} \frac{\partial B}{\partial
\kappa} - \frac{t \bar{v} \rho i u}{\sigma}\right)
\end{equation}
to reach the expression \eqref{eq:dphi_dkap}.
Similarly, Eq.~\eqref{eq:dphi_dsig} can be obtained by applying the chain rule to Eq.~\eqref{eq:SebCharFn}, and the intermediate terms for $\partial  \phi(\bm{\theta}; u, T)/\partial \sigma$ can be written in terms of those for $\partial  \phi(\bm{\theta}; u, T)/\partial \rho$, that is
\begin{subequations}\label{eq:fragmentsVov}
\begin{align}
\frac{\partial d}{\partial \sigma} &= \left(\frac{\rho}{\sigma} - \frac{1}{\xi}\right) \frac{\partial d}{\partial \rho} + \frac{\sigma u^2}{d},\\
\frac{\partial A_1}{\partial \sigma} &= \frac{(u^2+iu) t}{2} \frac{\partial d}{\partial \sigma}\cosh\frac{dt}{2},\\
\frac{\partial A_2}{\partial \sigma} &= \frac{\rho}{\sigma}\frac{\partial A_2}{\partial \rho} - \frac{2+t\xi}{iut\xi}\frac{\partial A_1}{\partial \rho} + \frac{\sigma tA_1}{2},\\
\frac{\partial A}{\partial \sigma} &= \frac{1}{A_2} \frac{\partial A_1}{\partial \sigma} - \frac{A}{A_2} \frac{\partial A_2}{\partial \sigma}.
\end{align}
\end{subequations}
%The intermediate terms are obtained via tedious derivation and will be omitted here.
%Some special care is needed to rearrange and merge terms in order to obtain the neat result of \eqref{eq:dC_dtheta}.
In the end, we replace $\log B$ appearing in Eqs.~\eqref{proof:dphi_dvbar} and \eqref{proof:dphi_dkappa} with $D$, defined in Eq.~\eqref{eq:D}, to ensure the numerical continuity of the implementation.
\end{proof}

Next we discuss the computation of the integrands in Eq.~\eqref{eq:dC_dtheta} and their convergence.

\subsubsection{Efficient calculation and convergence of the integrands}
All integrands have the form $\mathrm{Re} \left(\phi(\bm{\theta}; u, t)h_j(u)K^{-iu}/(iu)\right)$ and $h_j(u)$ is a product of elementary functions depending on which
parameter is under consideration.
It has been pointed out in the original paper by Heston \cite{Heston} that the term $\mathrm{Re} \left(\phi(\bm{\theta}; u, t)K^{-iu} /(iu)\right)$ is a smooth function that decays rapidly and presents no difficulties; its product with elementary functions decreases fast too.
A visual example is shown in Fig.~\ref{fig:convergeIntegrand}, with parameters given in Table~\ref{tab:parameter_exp1}. In our time units, $t=1$ is a trading year made of 252 days.
\begin{figure}[H]
    \subfloat[$T=60$ days.]{\includegraphics[width=0.48\textwidth]{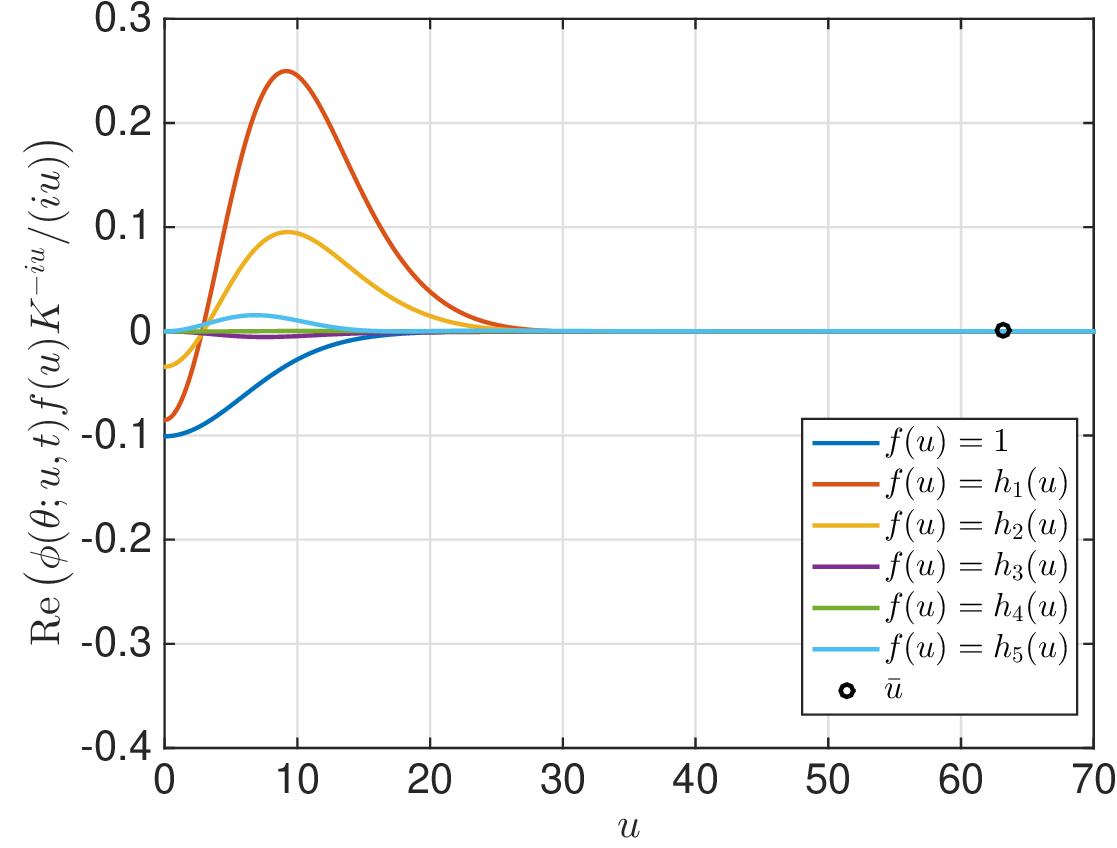}\label{fig:conv60}}\hfill
   \subfloat[$T=90$ days.]{\includegraphics[width=0.48\textwidth]{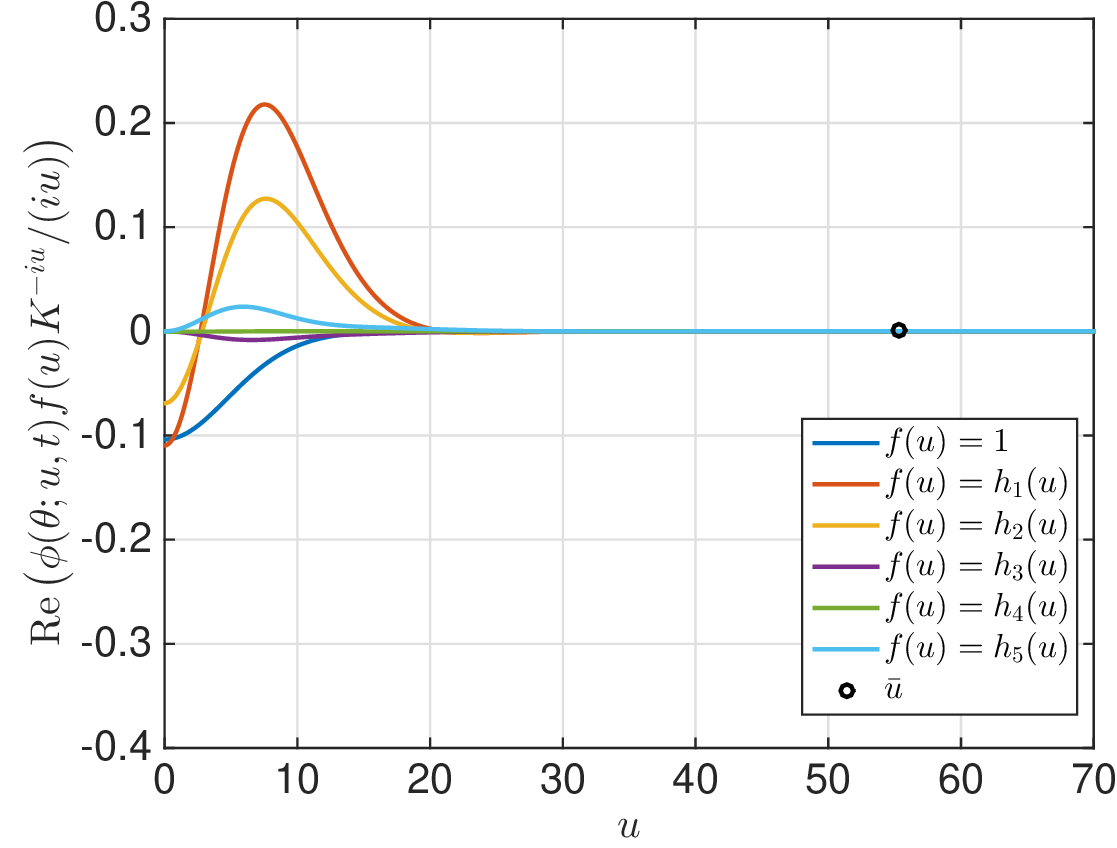}\label{fig:conv90}}\\
   % \end{figure}
   % \begin{figure}
   % \ContinuedFloat
     \subfloat[$T=180$ days.]{\includegraphics[width=0.48\textwidth]{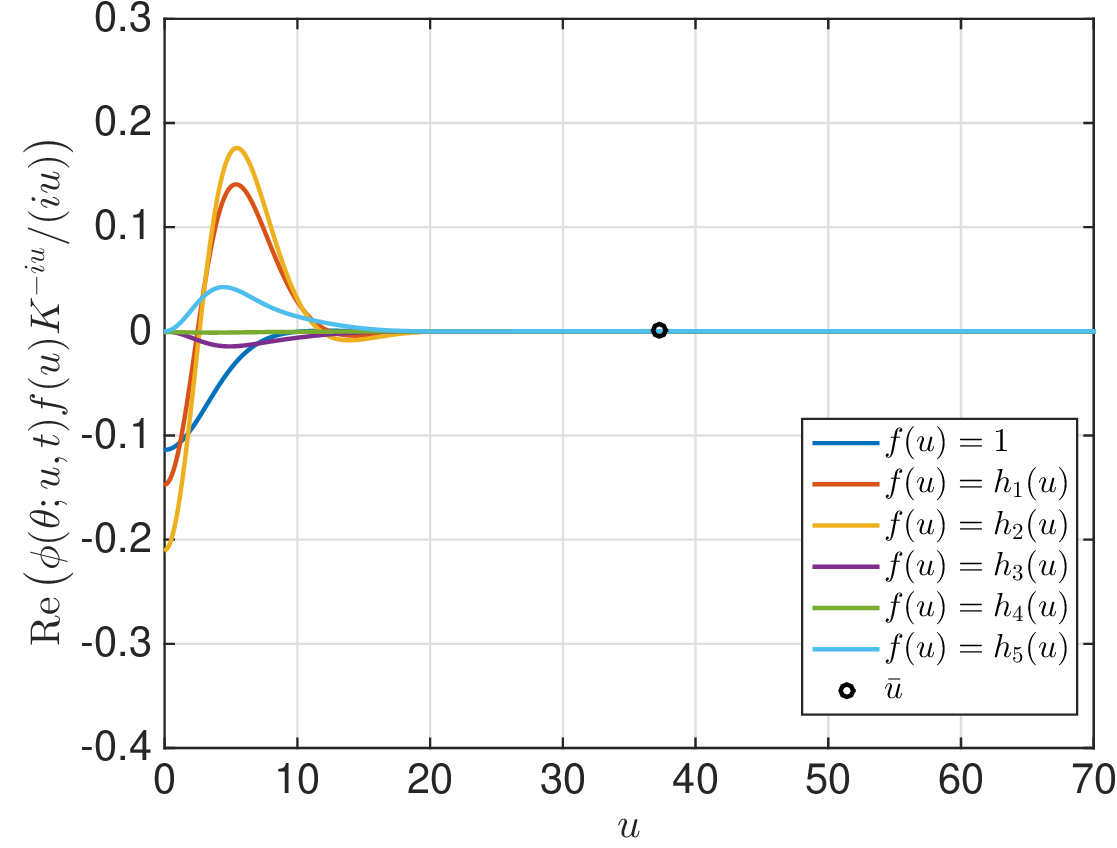}\label{fig:conv180}}\hfill
    \subfloat[$T=252$ days.]{\includegraphics[width=0.48\textwidth]{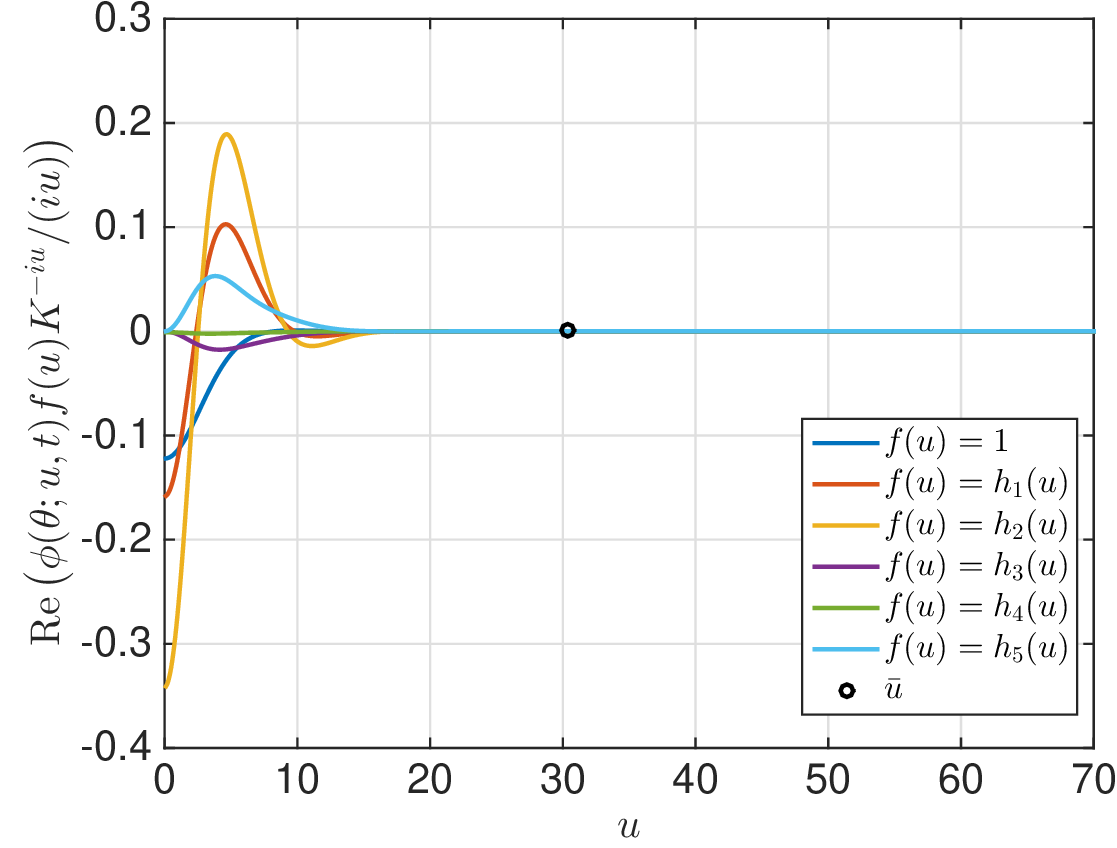}\label{fig:conv252}}
\caption{Convergence of the integrands:  $\mathrm{Re} \left(\phi(\bm{\theta}; u, t)K^{-iu}/(iu)\right)$ in the Heston pricing formula for $C(\bm{\theta}; K, T)$ (dark blue) and $\mathrm{Re} \left(\phi(\bm{\theta}; u, t)\bm{h}(u)K^{-iu}/(iu)\right)$ in the components of its gradient $\partial C/ \partial \bm{\theta}$ (other colors); $h_j(u), j = 1, \ldots, 5$ are respectively relevant for $\partial C/ \partial \theta_j$. The black circle indicates the value $\bar{u}$ where all integrands are below $10^{-8}$.\label{fig:convergeIntegrand}}
\end{figure}

Denote as $\bar{u}$ the value of $u$ for which all integrands are not larger than $10^{-8}$.  For our testing parameter set, we observe in Figs.~\ref{fig:convergeIntegrand} and \ref{fig:convpoint} that $\bar{u}$ decreases when $T$ increases. This is due to the fact that the more spread-out a function is, the more localised its Fourier transform is (see the uncertainty principle in physics): as $T$ increases, the probability density of $S_T$ stretches out, while its Fourier transform $\phi(\bm{\theta}; u, T)$ squeezes. More specifically, if $X$ and $U$ are random variables whose probability density functions are, apart of a constant, Fourier pairs of each other, the product of their variances is a constant, i.e., $\mathrm{Var}(X)\mathrm{Var}(U) \geq 1$.
Based on this observation, one can adjust the truncation according to the maturity of the option and hence do fewer integrand evaluations for options with longer maturities.
\begin{figure}[H]
\centering
%\vspace{+1cm}
    \includegraphics[width=0.48\textwidth]{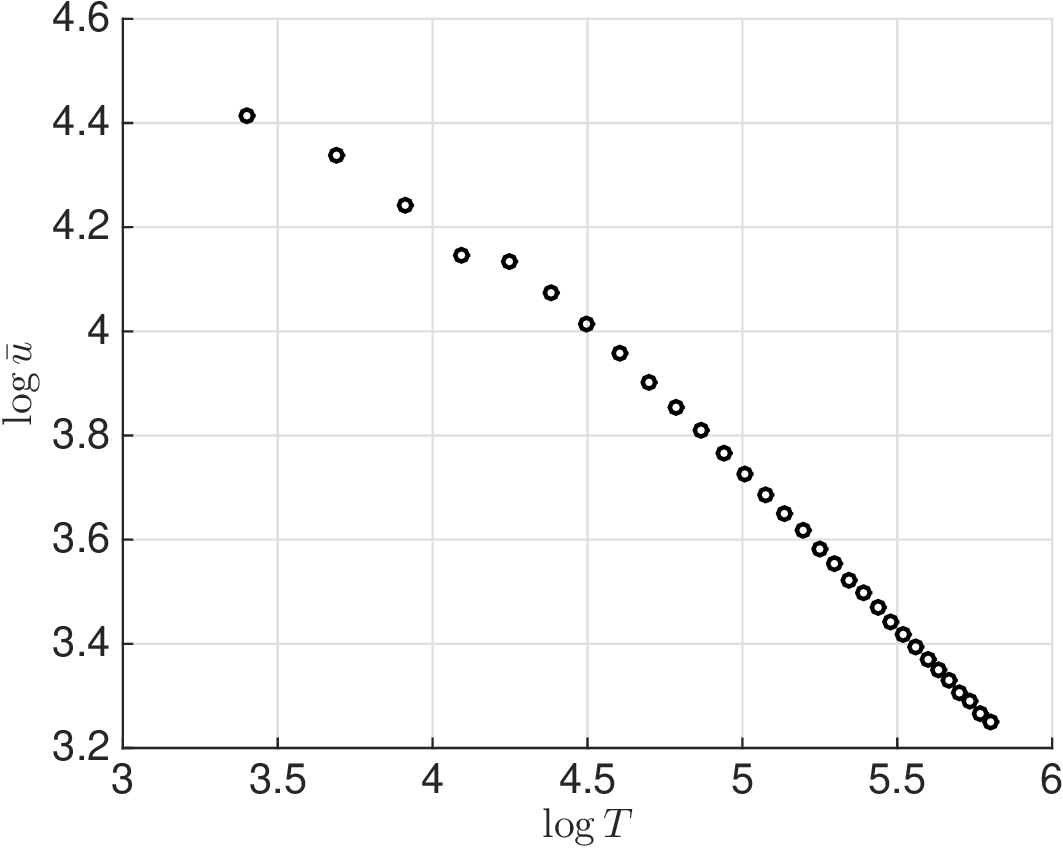}
\caption{As the maturity $T$ increases, the value $\bar{u}$ for which all integrands evaluate to $10^{-8}$ or less decreases.\label{fig:convpoint}}
\end{figure}

In order to obtain the integrands in Eq.~\eqref{eq:dC_dtheta}, one only needs to compute $\phi(\bm{\theta}; u, t)$ and $\bm{h}(u)$.
After rearranging and merging terms, we find that calculating $\bm{h}(u)$ can be boiled down to obtaining the intermediate terms \eqref{eq:fragmentsRho}, \eqref{eq:fragmentsKap} and \eqref{eq:fragmentsVov}.
%\begingroup
%\allowdisplaybreaks
%\begin{subequations}\label{eq:fragmentsInGradients}
%\begin{align}
%\frac{\partial d}{\partial \rho} &= -\frac{\xi \sigma i u}{d},\\
%\frac{\partial A_2}{\partial \rho} &= -\frac{\sigma i u (2+t\xi)}{2d}\left(\xi \cosh \frac{d t}{2} + d \sinh \frac{d t}{2}\right),\\
%\frac{\partial B}{\partial \rho}&= e^{\kappa t/2}\left(\frac{1}{A_2}\frac{\partial d}{\partial \rho} - \frac{d}{A_2^2}\frac{\partial A_2}{\partial \rho}\right),\\
%\frac{\partial d}{\partial \sigma} &= \left(\frac{\rho}{\sigma} - \frac{1}{\xi}\right) \frac{\partial d}{\partial \rho} + \frac{\sigma u^2}{d},\\
%\frac{\partial A_2}{\partial \sigma} &= \frac{\rho}{\sigma}\frac{\partial A_2}{\partial \rho} - \frac{2+t\xi}{iut\xi}\frac{\partial A_1}{\partial \rho} + \frac{\sigma tA_1}{2},\\
%\frac{\partial A_1}{\partial \rho} &= - \frac{i u(iu+u^2)t\xi \sigma }{2d}\cosh \frac{dt}{2} ,\\
%\frac{\partial A}{\partial \rho} &= \frac{1}{A_2} \frac{\partial A_1}{\partial \rho} - \frac{A}{A_2} \frac{\partial A_2}{\partial \rho},\\
%\frac{\partial B}{\partial \kappa}&=  \frac{i}{\sigma u} \frac{\partial B}{\partial \rho} + \frac{tB}{2},\\
%\frac{\partial A_1}{\partial \sigma} &= \frac{(iu+u^2) t}{2} \frac{\partial d}{\partial \sigma}\cosh\frac{dt}{2},\\
%\frac{\partial A}{\partial \sigma} &= \frac{1}{A_2} \frac{\partial A_1}{\partial \sigma} - \frac{A}{A_2} \frac{\partial A_2}{\partial \sigma}.
%\end{align}
%\end{subequations}
%\endgroup
It is a favorable result that the components of $\bm{h}(u)$ share these common terms because then the gradient $\bm{\nabla} C(\bm{\theta}; K,T)$ can be obtained by vectorizing the quadrature for all the integrands as illustrated in Algorithm \ref{algo:vectorizedIntegral}.
\begin{algorithm}[h]
\caption{Vectorised integration in the Heston gradient.}\label{algo:vectorizedIntegral}
Specify $N$ grid nodes $\left\{u_k\right\}_{k = 1} ^N$ and $N$ corresponding weights $\left\{w_k\right\}_{k = 1} ^N$.\\
\For{$k=1, 2, \dots, N$} {
            Compute $\bm{h}(u_k)$.\\
            }
\For{$j= 1, 2, \dots, 5$} {
            Compute $\int_0^\infty\frac{K^{-iu}}{iu}\phi(\bm{\theta}; u_k, t) h_j(u) \mathrm{d}u \approx \sum_{k=1}^{N}\frac{K^{-iu_k}}{iu_k}\phi(\bm{\theta}; u_k, t)h_j(u_k)w_k$.
            }
\end{algorithm}
Due to the interdependence among components of $\bm{h}(u)$, this scheme is faster than computing and integrating each component $h_j(u)$ individually.
Next, we discuss the choice of the numerical integration method and of the key parameters $N$, $u_k$ and $w_k$, but we point out that this vectorised quadrature is compatible with any numerical integration method.

\subsubsection{Integration scheme}
The computation of the integrals in the pricing function \eqref{eq:HestonCall} and the gradient function \eqref{eq:dC_dtheta} dominates the cost of calibration.
Thus, we discuss the proper choice of the numerical integration scheme.
Specifically, we compare the trapezoidal rule (TR) and the Gauss-Legendre rule (GL).
In Figs.~\ref{fig:node100_pv} and \ref{fig:node100_jac}, we plot the error of the integral evaluation respectively in the pricing formula and its gradient.
\begin{figure}[h]\label{fig:compareRule}
\centering
    \subfloat[for the pricing function]{\includegraphics[width=.48\textwidth]{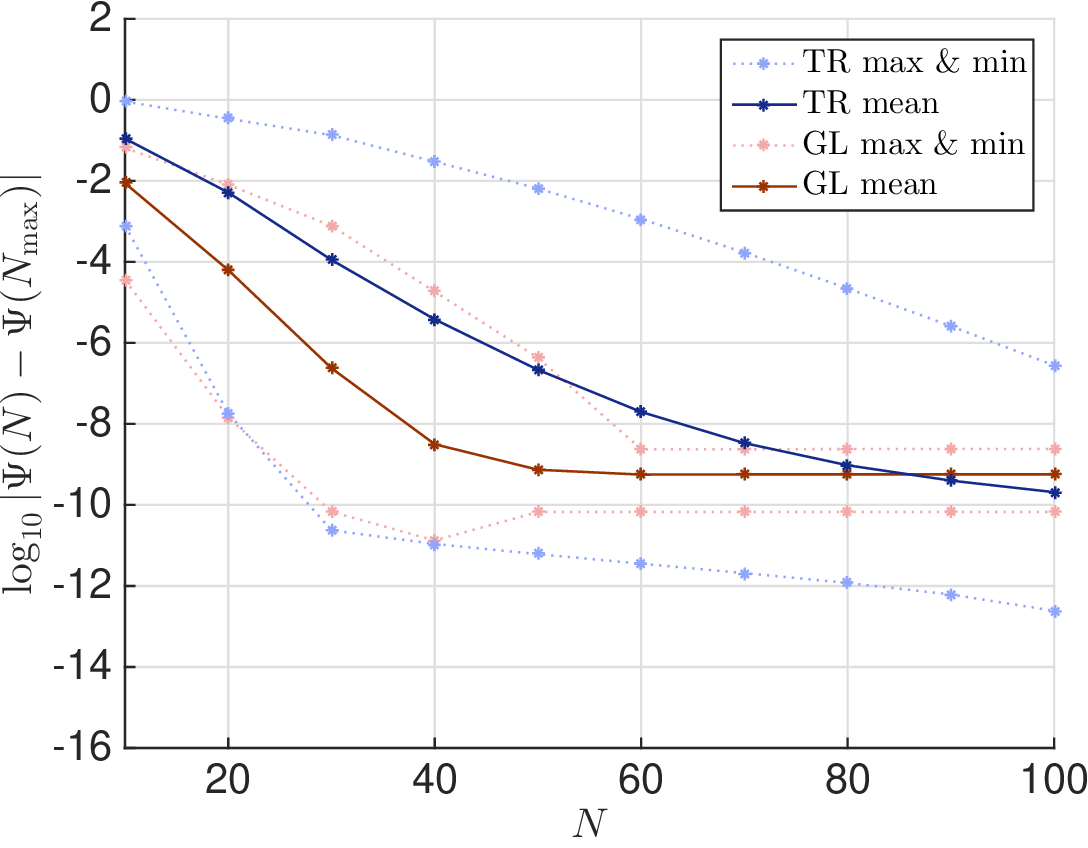}\label{fig:node100_pv}}
 %\hspace{2cm}
 \hfill
    \subfloat[for the gradient function]{\includegraphics[width=.48\textwidth]{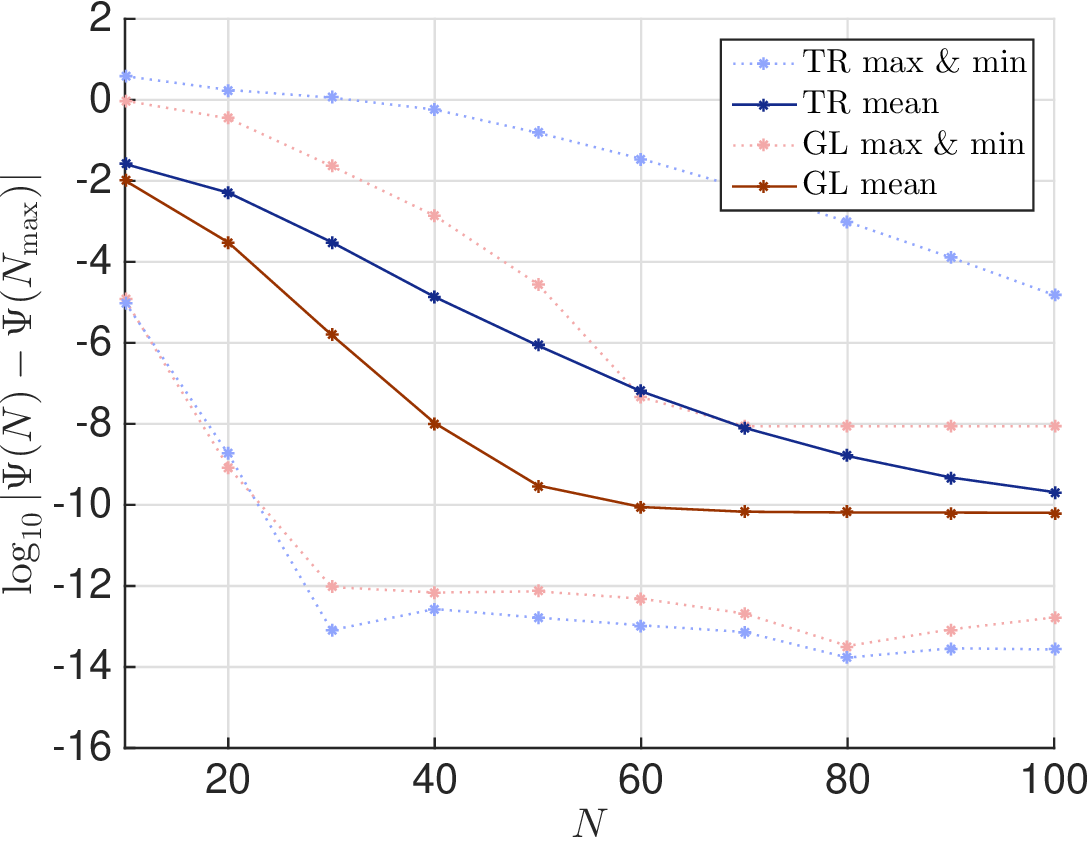}\label{fig:node100_jac}}
    \caption{Comparison between TR (full blue for mean and dotted purple for maximum and minimum) and GL (full red for mean and dotted pink for maximum and minimum) for the error of the integral evaluation under the Heston model.}
\end{figure}
The horizontal axis is the number of quadrature nodes $N$ and the vertical axis is the $\log_{10}$ scale of the absolute error of integration, which is defined as
\begin{equation}
\varepsilon_{\mathrm{integration}} := | \Phi(N) - \Phi(N_{\mathrm{max}})|,
\end{equation}
where $\Phi(N)$ is the value of the integration with $N$ nodes, $N$ is selected equidistantly in the range $[10, 100]$, $N_{\mathrm{max}}$
should be $\infty$ and is chosen as $1000$ in our case.
For the plots we use $40$ options with different strikes and maturities.
More details on these options are given in Section~\ref{sec:results}.
%The lines in shades of blue are computed using TR and those in shades of red are computed using GL.
%Precisely, we plot the mean in full blue (red) lines and the maximum and minimum in dotted purple (light pink) lines.

The error converges faster for GL than TR and has always a smaller variation when more options are involved.
In order to achieve an average accuracy of $10^{-8}$, GL requires $40$ nodes and TR requires $70$.
In order for the integrations for all the options to achieve an accuracy at $10^{-8}$, GL requires $60$ nodes and TR requires much more than $100$.

Besides the fast convergence of the integral error, GL is advantageous in its selection of nodes.
GL rescales the domain of integration to $[-1, +1]$, selects nodes that are symmetric around the origin, and assigns the same weight to each symmetric pair of nodes.
Thus, a further reduction in computation can be achieved by making use of the common terms of a node and its opposite.
Based on these benefits, we choose the GL integration scheme with about $60$ nodes to calibrate the Heston model.

\subsubsection{Comparison with numerical gradient}
Previous calibration methods approximate the gradient by a finite difference scheme. A central difference scheme is the approximation
%partial derivative at $\theta_j$ as
\begin{equation}
\bm{\nabla}C(\bm{\theta}; K, T) \approx \frac{C(\bm{\theta}+\bm{\epsilon}; K, T)-C(\bm{\theta}-\bm{\epsilon};K, T)}{2\epsilon},
%\frac{\partial C(\bm{\theta}; K_i, T_i)}{\partial \theta_j} \approx \frac{C(\bm{\theta}+\epsilon\bm{e}_j; K_i, T_i)-C(\bm{\theta}-\epsilon\bm{e}_j;K_i, T_i)}{2\epsilon},
\end{equation}
where $\bm{\epsilon} := \epsilon\bm{e}$ and $\epsilon$ is small. Different values of the increment $\epsilon$ could be chosen for each component $\theta_j$; for simplicity we have taken it constant.
The size of the difference, $\epsilon$, has a non-trivial effect on the approximation.
An excessively small value of $\epsilon$ is not able to reflect the overall function behavior at the point and may lead to a wrong moving direction.
Moreover the numerical gradient naturally has an error and one cannot expect to find a solution with a better accuracy than that of the gradient.
In most cases, the iteration stagnates when the error of the objective function is roughly the same size as the error of the gradient.

Besides the instability caused by an inappropriate choice of $\epsilon$, a numerical gradient has a higher computational cost than an analytical gradient.
Recall that the evaluation of one option price $C(\bm{\theta}; K, T)$ requires the evaluation of two integrals as in Eq.~\eqref{eq:HestonCall}.
Let $n$ be the number of options to be calibrated.
At each iteration, one needs to compute $20n$ integrals if using the finite difference scheme while only $2n$ integrals if using the analytical form with the vectorised integration scheme.
To give a more intuitive comparison between the two methods, we perform a preliminary experiment with $\epsilon = 10^{-4}$ and $n = 40$ using the MATLAB function \textit{quadv} with an adaptive Simpson rule for the numerical integration.
In Table~\ref{tab:experimentOnFiniteDifference}, we report the CPU time as an average of $500$ runs and the number of calls of the integral function for each method.
In order to give a relative sense of speed that is independent of the machine, the CPU time for analytical gradient is scaled to unity, and that for numerical gradient results about $16$ times longer.
\begin{table}[H]
\caption{A comparison between numerical and analytical gradients for $n=40$ options.}\label{tab:experimentOnFiniteDifference}
\begin{center}\footnotesize
\renewcommand{\arraystretch}{1.3}
\begin{tabular}{l |rr}\hline\hline
\multicolumn{1}{c|}{Computational cost}       & \multicolumn{1}{c}{Numerical gradient} & \multicolumn{1}{c}{Analytical gradient} \\ \hline
 CPU time (arbitrary units)               & 15.8                      & 1.0 \\
 Number of integral evaluations      & 800                         & 80 \\ \hline\hline
\end{tabular}
\end{center}
\end{table}
Considering the $94\%$ of saving in computational time and the exempt from deciding $\epsilon$, we propose to use the analytical Heston gradient with vectorised quadrature in a gradient-based optimisation algorithm to calibrate the model.

\section{Calibration using the Levenberg-Marquardt method}\label{sec:calibration}
In this section, we present the algorithm for a complete and fast calibration of the Heston model using the LM method \cite{more1978levenberg}.

The LM method is a typical tool to solve a nonlinear least squares problem like Eq.~\eqref{eq:LSform}.
The search step is given by
\begin{equation}\label{eq:LMnormalequation}
\Delta \bm{\theta} = (\bm{J}\bm{J}^\intercal  + \mu\bm{I})^{-1}\bm{\nabla}f,
\end{equation}
where $\bm{I}$ is the identity matrix and $\mu$ is a damping factor. By adaptively adjusting $\mu$, the method changes between the steepest descent method and the Gauss-Newton method: when the iterate is far from the optimum, $\mu$ is given a large value so that the Hessian matrix is dominated by the scaled identity matrix
\begin{equation} \label{eq:HessianSD}
\bm{\nabla}\bm{\nabla}^\intercal f \approx \mu\bm{I};
\end{equation}
when the iterate is close to the optimum, $\mu$ is assigned a small value so that the Hessian matrix is dominated by the Gauss-Newton approximation
\begin{equation} \label{eq:HessianGN}
\bm{\nabla}\bm{\nabla}^\intercal f \approx \bm{J}\bm{J}^\intercal,
\end{equation}
which omits the second term $\sum^n_{i=1} r_i \bm{H}(r_i)$ in Eq.~\eqref{HessianLS}.
The approximation \eqref{eq:HessianGN} is reliable when either $r_i$ or $\bm{H}(r_i)$ is small.
The former happens when the problem is a so-called \textit{small residual problem} and the latter happens when $f$ is nearly linear.
The viewpoint is that the model should yield small residuals around the optimum because otherwise it is an inappropriate model.
The Heston model has been known to be able to explain the smile and skew of the volatility surface.
Therefore, we conjecture it to be a small residual problem and adopt the approximation of the Hessian in Eq.~\eqref{eq:HessianGN} as converging to the optimum.
There are various implementations of the LM method, such as \texttt{MINPACK} \cite{cminpack}, \texttt{LEVMAR} \cite{lourakis2005brief}, \texttt{sparseLM} \cite{lourakis10} etc.
We adopt the {\tt{LEVMAR}} package which is a robust and stable implementation in C/C++ distributed under GNU.
Although its documentation does not report a use in computational finance, \texttt{LEVMAR} has been integrated into many open source and commercial products in other applications such as astrometric calibration and image processing. See Algorithm \ref{algo:levmar}.
\begin{algorithm}[h]
\caption{Levenberg-Marquardt algorithm to calibrate the Heston model.}\label{algo:levmar}
Given the initial guess $\bm{\theta}_0$, compute $\|\bm{r}(\bm{\theta}_0)\|$ and $\bm{J}_0$.\label{algo:feva1}\\
Choose the initial damping factor as $\mu_0 = \tau \max\left\{\mathrm{diag}(\bm{J}_0)\right\}$ and  $\nu_0 = 2$.\\
\For{$k=0, 1, 2, \ldots$} {

Solve the normal equations \eqref{eq:LMnormalequation} for $\Delta \bm{\theta}_k$.\label{algo:levmarNE}\\
Compute $\bm{\theta}_{k+1} = \bm{\theta}_k + \Delta \bm{\theta}_k$ and $\|\bm{r}(\bm{\theta}_{k+1})\|$.\label{algo:feva2}\\
Compute $\delta_L= \Delta{\bm{\theta}_k}^\intercal(\mu\Delta\bm{\theta}_k + {\bm{J}_k}\bm{r}(\bm{\theta}_k)) $ and $\delta_F = \|\bm{r}(\bm{\theta}_k)\| - \|\bm{r}(\bm{\theta}_{k+1})\|$.\\
\eIf{$\delta_L>0$ and $\delta_F>0$}{
	Accept the step: compute $\bm{J}_{k+1}$, $\mu_{k+1} = \mu_k$, $\nu_{k+1} = \nu_k$.\label{algo:jaceva2}
}{
	Recalculate the step: set $\mu_{k} = \mu_k\nu_k$, $\nu_{k} = 2\nu_k$ and repeat from line \ref{algo:levmarNE}.
}
\If{the stopping criterion \eqref{eq:LMstop} is met}{
Break.
}
        }
\end{algorithm}

In lines \ref{algo:feva1} and \ref{algo:feva2} of Algorithm \ref{algo:levmar}, the option pricing function must be evaluated.
In lines \ref{algo:feva1} and \ref{algo:jaceva2}, the gradient function needs to be evaluated.
In line \ref{algo:levmarNE}, a $5\times 5$ linear system needs to be solved;
in {\tt{LEVMAR}} this is done by an LDLT factorization with the pivoting strategy of Bunch and Kaufman \cite{bunch1976decomposition} using the {\tt{LAPACK}} \cite{anderson1999}
routine.

The stopping criterion for the LM algorithm is when one of the following is satisfied:
\begin{subequations}\label{eq:LMstop}
\begin{align}
\|\bm{r}(\bm{\theta}_k)\| &\leq \varepsilon_1,\label{eq:stop1}\\
\|{\bm{J}_k}\bm{e}\|_\infty&\leq \varepsilon_2,\label{eq:stop2}\\
\frac{\|\Delta\bm{\theta}_k\|}{\|\bm{\theta}_k\|}&\leq \varepsilon_3\label{eq:stop3},
\end{align}
\end{subequations}
where $\varepsilon_1, \varepsilon_2$ and $\varepsilon_3$ are tolerance levels.
The first condition \eqref{eq:stop1} indicates that the iteration is stopped by a desired value of the objective function \eqref{eq:LSform}--\eqref{eq:defOfLS}.
The second condition \eqref{eq:stop2} indicates that the iteration is stopped by a small gradient.
The third condition \eqref{eq:stop3} indicates that the iteration is stopped by a stagnating update.

\section{Numerical results}\label{sec:results}
In this section, we present our experimental results for the calibration of the Heston model.
We first describe the data and then report the performance of our calibration method in comparison with the fastest previous method.
We examine the Hessian matrix at the optimal solution which reveals the reason of the multiple optima observed in previous research.
In the end, we test on three parameterisations that are typical for certain options. The result justifies the computational efficiency and robustness of our method for practical problems.

\subsection{Data}
In order to check whether the optimal parameter set found by the algorithm is the global optimum, we first presume a parameter set $\bm{\theta}^*$ specified in Table~\ref{tab:parameter_exp1}, and then use it to generate a volatility surface that is typically characterised by these options: the $\Delta_{10}$ call and put options, $\Delta_{25}$ call and put options, and $\Delta_{50}$ (i.e., ATM) call options with maturity from $30$ to $360$ days.
Here $\Delta:=\partial C(\bm{\theta}; K, T) /\partial S$ is the BS greek, i.e., the sensitivity of the option price with respect to the movement of its underlying spot.
In Table~\ref{tab:volsurface}, we give the BS implied volatilities of $40$ options that are generated by $\bm{\theta}^*$.
We denote call and put options using superscripts, respectively as $\Delta^{\mathrm{call}}$ and $\Delta^{\mathrm{put}}$.
\begin{table}[h]
\caption{Volatility surface for calibration.}\label{tab:volsurface}
\begin{center}\footnotesize
\renewcommand{\arraystretch}{1.3}
\begin{tabular}{r |rrrrr}\hline\hline
Maturity in days & $\Delta_{10}^{\mathrm{put}}$ & $\Delta_{25}^{\mathrm{put}}$ & $\Delta_{50}^{\mathrm{call}}$  & $\Delta_{25}^{\mathrm{call}}$ & $\Delta_{10}^{\mathrm{call}}$\\
\hline
 30 & 2.5096    & 1.4359   & 0.2808 & 0.2540 & 0.2369 \\
 60 & 2.4351    & 1.3216   & 0.2847 & 0.2606 & 0.2417 \\
 90 & 2.3823    & 1.2955   & 0.2878 & 0.2660 & 0.2489 \\
 120 & 2.3383   & 1.2677   & 0.2904 & 0.2699 & 0.2548 \\
 150 & 2.2996   & 1.2407   & 0.2925 & 0.2745 & 0.2598 \\
 180 & 2.2619   & 1.2166   & 0.2943 & 0.2777 & 0.2641 \\
 252 & 2.1767   & 1.1671   & 0.2975 & 0.2837 & 0.2722 \\
 360 & 2.0618   & 1.1136   & 0.3007 & 0.2897 & 0.2803 \\ \hline\hline
\end{tabular}
\end{center}
\end{table}
The target is thus to find a parameter set $\bm{\theta}^\dag$ that can replicate the volatility surface in Table~\ref{tab:volsurface}.
If $\bm{\theta}^\dag$ is far from $\bm{\theta}^*$ or in other words, depends on the initial guess $\bm{\theta}_0$, then one concludes that local optimal parameter sets exist.
Otherwise the problem presents only a global optimum.

We validated our method using different optimal parameters and initial guesses in a reasonable range given in Table~\ref{tab:parameter_initialpoint}. The procedure is described in Algorithm \ref{algo:test}.
\begin{table}[h]
\caption{Reasonable ranges to randomly generate Heston model parameters and the average absolute distance between the initial guess $\bm{\theta}_0$ and the optimum $\bm{\theta}^*$.}
\begin{center}\footnotesize
\renewcommand{\arraystretch}{1.3}
\begin{tabular}{cr|cr}\hline\hline
   % \multicolumn{2}{c}{ Model parameters} \\ \hline
     \multicolumn{2}{c} { Range for model parameters} &  \multicolumn{2}{c}{ Absolute deviation from $\bm{\theta}^*$ } \\ \hline
     $\kappa$               & $(0.50, 5.00)$     &$|\kappa_0- \kappa^*|$  & 1.5097\\
     $\bar{v}$               & $(0.05, 0.95)$     & $|\bar{v}_0-\bar{v}^*|$   & 0.2889\\
     $\sigma$               & $(0.05, 0.95)$     &$|\sigma_0-\sigma^*|$   & 0.2875\\
     $\rho$                   & $(-0.90, -0.10)$   &$|\rho_0-\rho^*|$            & 0.2557\\
     $v_0$                   & $(0.05, 0.95)$      & $|{(v_0)}_0-v_0^*|$       & 0.3063\\ \hline\hline
\end{tabular}\label{tab:parameter_initialpoint}
\end{center}
\end{table}

\begin{algorithm}[H]
\caption{Validation procedure.}\label{algo:test}
%Independently generate 100 sets of initial guesses $\left\{\bm{\theta}_0\right\}_{k=1}^{100}$, each dimension of $\left\{\bm{\theta}_0\right\}_{k}$ is a uniformly distributed random number in the interval specified in Table \ref{tab:parameter_initialpoint}.\\
%Similarly generate 100 sets of optimal parameters $\left\{\bm{\theta}^*\right\}_{l=1}^{100}$.\\
\For{$i=1, 2, \ldots, 100$} {
Generate a vector of optimal parameters $\bm{\theta}^*_i$, each component of which is an independent uniformly distributed random number in the interval specified in Table~\ref{tab:parameter_initialpoint}.\\
    \For{$j=1, 2, \ldots, 100$} {
Generate an initial guess $\bm{\theta}_{0j}$, each component of which is an independent uniformly distributed random number in the interval specified in Table~\ref{tab:parameter_initialpoint}.\\
    Validate Algorithm \ref{algo:levmar} using the initial guess $\bm{\theta}_{0j}$ to find $\bm{\theta}^*_i$.
    }
}
\end{algorithm}
Following this procedure, we validated Algorithm \ref{algo:levmar} with 10\,000 test cases. An average of the distances between the initial guesses $\bm{\theta}_0$ and the optima $\bm{\theta}^*$ is given in Table~\ref{tab:parameter_initialpoint}. The results of the tests are discussed in the next section.

\subsection{Performance}
The computations were performed on a MacBook Pro with a 2.6 GHz Intel Core i5 processor, 8 GB of RAM and OS X Yosemite version 10.10.5.
The pricing and gradient functions for the Heston model were coded in C++ using Xcode version 7.3.1.
We use {\tt{LEVMAR}} version 2.6 \cite{lourakis2005brief} as the LM solver setting the tolerances in Eqs.~\eqref{eq:LMstop} to $\varepsilon_1 = \varepsilon_2 = \varepsilon_3 = 10^{-10}$.
However, in our experiments the LM iteration was always stopped by meeting the condition on the objective function \eqref{eq:stop1}.
We use GL integration with $N = 64$ nodes and for simplicity we truncate the upper limit of the integration in Eq.~\eqref{eq:dC_dtheta} at $\bar{u} = 200$ which shall be enough for pricing and calibrating in all cases.
The code is provided in the supplementary material.

The proposed method succeeds in finding the presumed parameter set in %9\,636
9\,843 cases out of 10\,000 without any constraints on the search space and in %9\,899
9\,856 cases restraining the search to the intervals specified in Table~\ref{tab:parameter_initialpoint}. The average CPU time for the whole calibration process is less than 0.3 seconds.
%From all initial guesses $\bm{\theta}_0$ we converged to the global optimum $\bm{\theta}^*$.
See Table~\ref{tab:infooptimalavg} for detailed information on the whole validation set. In Table~\ref{tab:infooptimal} and in the rest of this section we specify the information for a representative example  with the optimal parameter set $\bm{\theta}^*$ specified in Table~\ref{tab:parameter_exp1} and the initial guess $\bm{\theta}_0 = [1.20, 0.20, 0.30, -0.60, 0.20]^\intercal$.

 \begin{table}[h]
 \caption{Information about the optimisation: average over 10\,000 testing cases.}
 \begin{center}\footnotesize
 \renewcommand{\arraystretch}{1.3}
 \begin{tabular}{ lr | lr  | lr}\hline\hline
 \multicolumn{2}{c}{ Absolute deviation from $\bm{\theta}^*$}  & \multicolumn{2}{c}{Error measure}    & \multicolumn{2}{c}{Computational cost}\\ \hline
 $|\kappa^\dag - \kappa^*|$       & $1.54\times 10^{-3}$       & $\|\bm{r}_0\|$       & $1.39\times 10^{-1}$    & CPU time (seconds) & $0.29$ \\
 $|\bar{v}^\dag-\bar{v}^*|$          & $2.40\times 10^{-5}$       & $\|\bm{r}^\dag\|$  & $2.94\times 10^{-11}$  & LM iterations & $12.82$ \\
 $|\sigma^\dag-\sigma^*|$         & $3.79\times 10^{-3}$       & $\|{\bm{J}^\dag}\bm{e}\|_\infty$ & $1.47\times 10^{-5}$ & price evaluations & $14.57$\\
 $|\rho^\dag-\rho^*|$                  & $1.52\times 10^{-2}$       & $\|\Delta\bm{\theta}^\dag\|$   &$3.21\times 10^{-4}$ & gradient evaluations & $12.82$\\
 $|v_0^\dag-v_0^*|$                  & $6.98\times 10^{-6}$        &                             &                           & linear systems solved & $13.57$\\ \hline\hline
 \end{tabular}\label{tab:infooptimalavg}
 \end{center}
 \end{table}

 \begin{table}[h]
 \caption{Information about the optimisation of a representative example.}
 \begin{center}\footnotesize
 \renewcommand{\arraystretch}{1.3}
 \begin{tabular}{lr | lr  | lr}\hline\hline
 \multicolumn{2}{c}{ Absolute deviation from $\bm{\theta}^*$}  & \multicolumn{2}{c}{Error measure}    & \multicolumn{2}{c}{Computational cost}\\ \hline
 $|\kappa^\dag - \kappa^*|$       & $1.09\times 10^{-3}$       & $\|\bm{r}_0\|$       & $4.73\times 10^{-2}$    & CPU time (seconds) & $0.29$ \\
 $|\bar{v}^\dag-\bar{v}^*|$          & $2.18\times 10^{-6}$       & $\|\bm{r}^\dag\|$  & $1.00\times 10^{-12}$  & LM iterations & $13$ \\
 $|\sigma^\dag-\sigma^*|$         & $4.70\times 10^{-5}$       & $\|{\bm{J}^\dag}\bm{e}\|_\infty$ & $1.21\times 10^{-5}$ & price evaluations & $14$\\
 $|\rho^\dag-\rho^*|$                  & $9.89\times 10^{-6}$       & $\|\Delta\bm{\theta}^\dag\|$   &$2.50\times 10^{-4}$ & gradient evaluations & $13$\\
 $|v_0^\dag-v_0^*|$                  & $1.18\times 10^{-6}$        &                             &                           & linear systems solved & $13$\\ \hline\hline
 \end{tabular}\label{tab:infooptimal}
 \end{center}
 \end{table}

%  \begin{table}[H]
% \caption{Information about the optimisation: initial guess 1, optimal parameter set 1.}
% \begin{center}\footnotesize
% \renewcommand{\arraystretch}{1.3}
% \begin{tabular}{lr | lr | lr  | lr}\hline\hline
% \multicolumn{4}{c}{ Absolute deviation from $\bm{\theta}^*$}  & \multicolumn{2}{c}{Error measure}    & \multicolumn{2}{c}{Computational cost}\\ \hline
%$|\kappa_0- \kappa^*|$        & $1.80$   & $|\kappa^\dag - \kappa^*|$       & $1.0901\times 10^{-3}$       & $\|\bm{r}_0\|$       & $4.7320\times 10^{-2}$    & CPU time (seconds) & $0.2997$ \\
%$|\bar{v}_0-\bar{v}^*|$          & $0.10$   & $|\bar{v}^\dag-\bar{v}^*|$          & $2.1780\times 10^{-6}$       & $\|\bm{r}^\dag\|$  & $1.0021\times 10^{-12}$  & LM iterations & $13$ \\
%$|\sigma_0-\sigma^*|$         & $0.05$   & $|\sigma^\dag-\sigma^*|$         & $4.6988\times 10^{-5}$       & $\|{\bm{J}^\dag}\bm{e}\|_\infty$ & $1.2079\times 10^{-5}$ & price evaluations & $14$\\
%$|\rho_0-\rho^*|$                  & $0.20$   & $|\rho^\dag-\rho^*|$                  & $9.8880\times 10^{-6}$       & $\|\Delta\bm{\theta}^\dag\|$   &$2.5009\times 10^{-4}$ & gradient evaluations & $13$\\
%$|{(v_0)}_0-v_0^*|$                   & $0.12$   & $|v_0^\dag-v_0^*|$                  & $1.1818\times 10^{-6}$        &                             &                           & linear systems solved & $13$\\ \hline\hline
% \end{tabular}\label{tab:infooptimal}
% \end{center}
% \end{table}

\begin{figure}[H]
\centering
\includegraphics[width=0.48\textwidth]{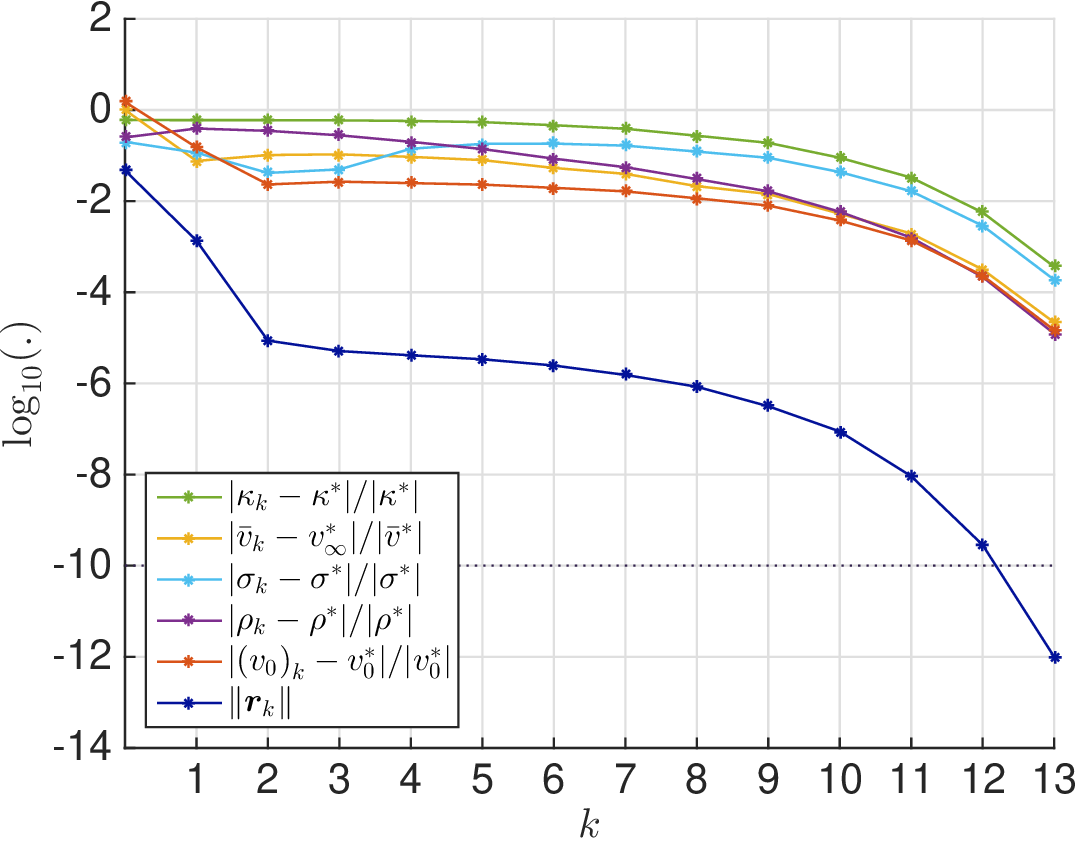}
\caption{The convergence of the LM method.\label{fig:convLM}}
\end{figure}
The convergence of the residual $\bm{r}_k$ and the relative distance of each parameter towards the optimum is plotted in Fig.~\ref{fig:convLM}.
In Figs.~\ref{fig:errsurf0} and \ref{fig:errsurfk}, we plot the pricing error on the implied volatility surface at the initial point $\bm{\theta}_0$ and the optimal point $\bm{\theta}^\dag$, respectively.
As can be seen, the pricing error decreases from $10^{-2}$ to $10^{-7}$ after $13$ steps.
\begin{figure}[H]
\centering
\subfloat[at the initial point $\bm{\theta}^0$]{\includegraphics[width=0.48\textwidth]{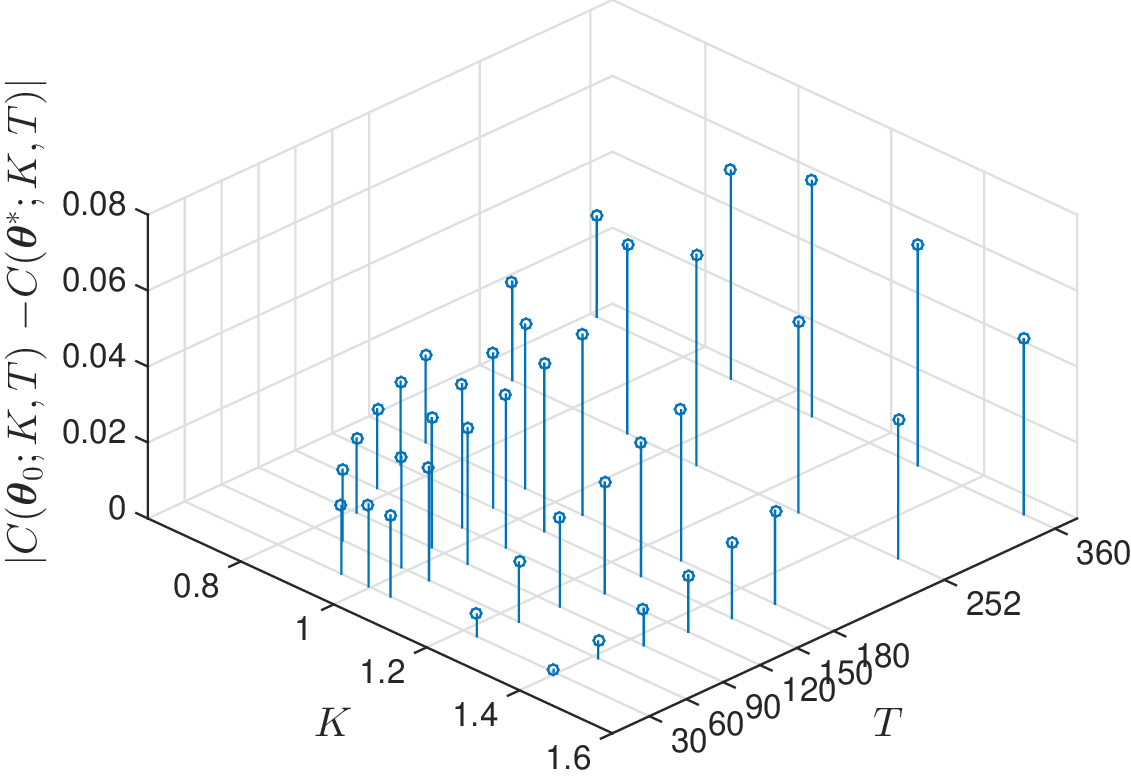}\label{fig:errsurf0}}\hfill
\subfloat[the final point $\bm{\theta}^\dag$]{\includegraphics[width=0.48\textwidth]{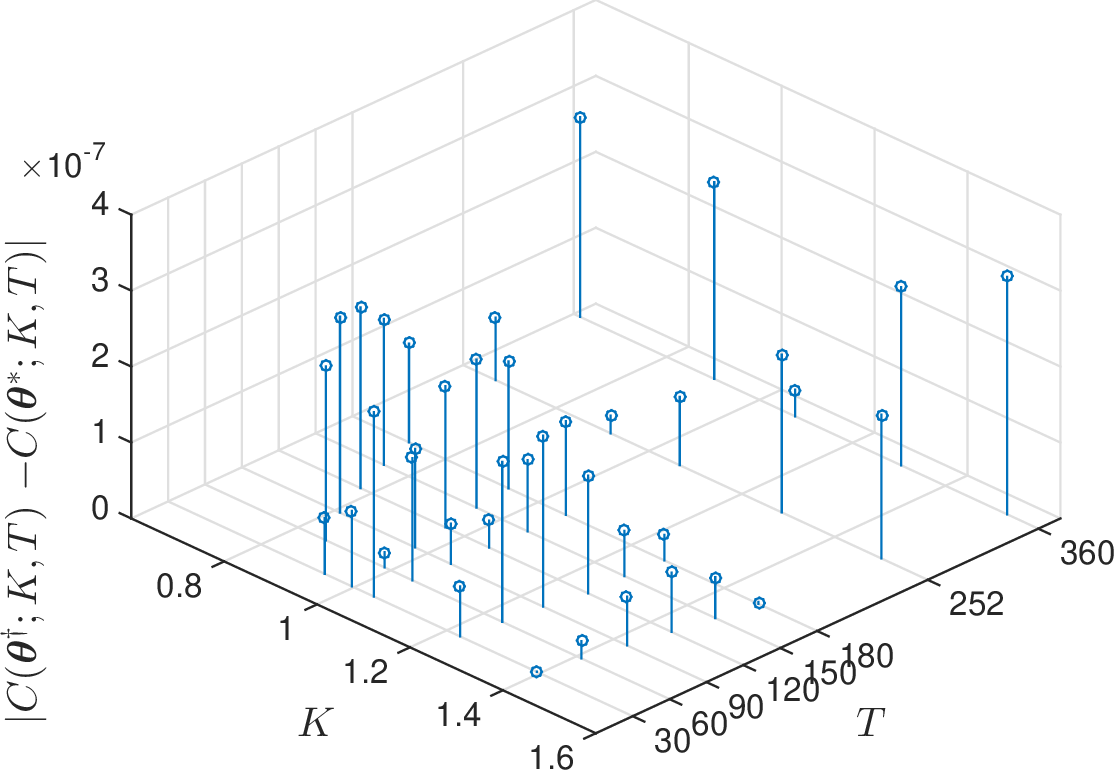}\label{fig:errsurfk}}
\caption{Pricing error on the implied volatility surface.}% at (a) the initial point $\bm{\theta}^0$ and (b) the final point $\bm{\theta}^\dag$.
\end{figure}
This result contrasts the conclusion of previous research:
local optimal parameters are not intrinsically embedded in the Heston calibration problem, but rather caused by an objective function shaped as a narrow valley with a flat bottom and a premature stopping criterion.

We plot the contours for $\|\bm{r}\|$ when varying $2$ out of $5$ parameters.
Starting from $\bm{\theta}_0$, the iteration path is shown with contour plots in Fig.~\ref{fig:contour}.
The initial point $\bm{\theta}_0$ is marked with a black circle and the true solution $\bm{\theta}^*$ is marked with a black plus symbol.
The red lines with asterisks are the iteration paths of $\bm{\theta}_k,\; k = 1, \dots, 13$.
For almost all pairs, the first step is a long steepest descent step that is nearly orthogonal to the contour.
The rest are relatively cautious steps with the Gauss-Newton approximation of the Hessian.
The contour plots do not show evidence for local minima, at least not in $2$ dimensional sections.
\begin{figure}[H]
\centering
\subfloat[]{\includegraphics[width=0.48\textwidth]{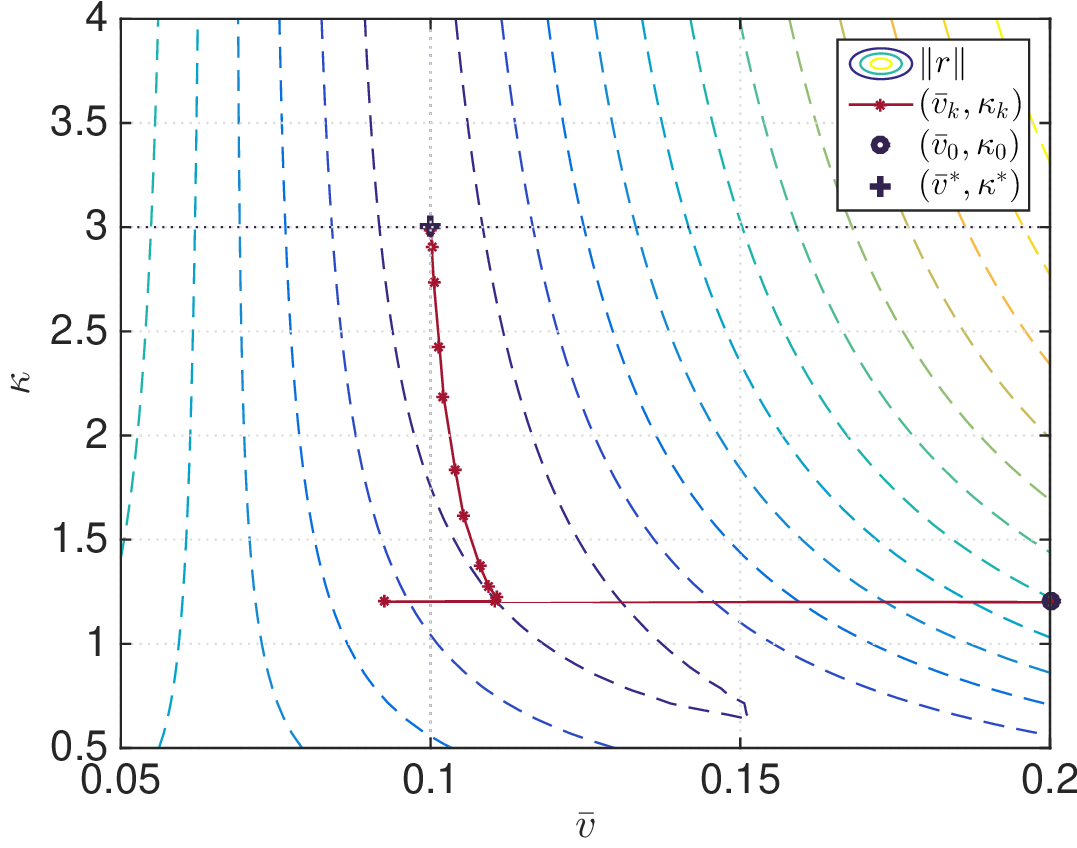}\label{fig:cont_kapvinf}}\hfill
\subfloat[]{\includegraphics[width=0.48\textwidth]{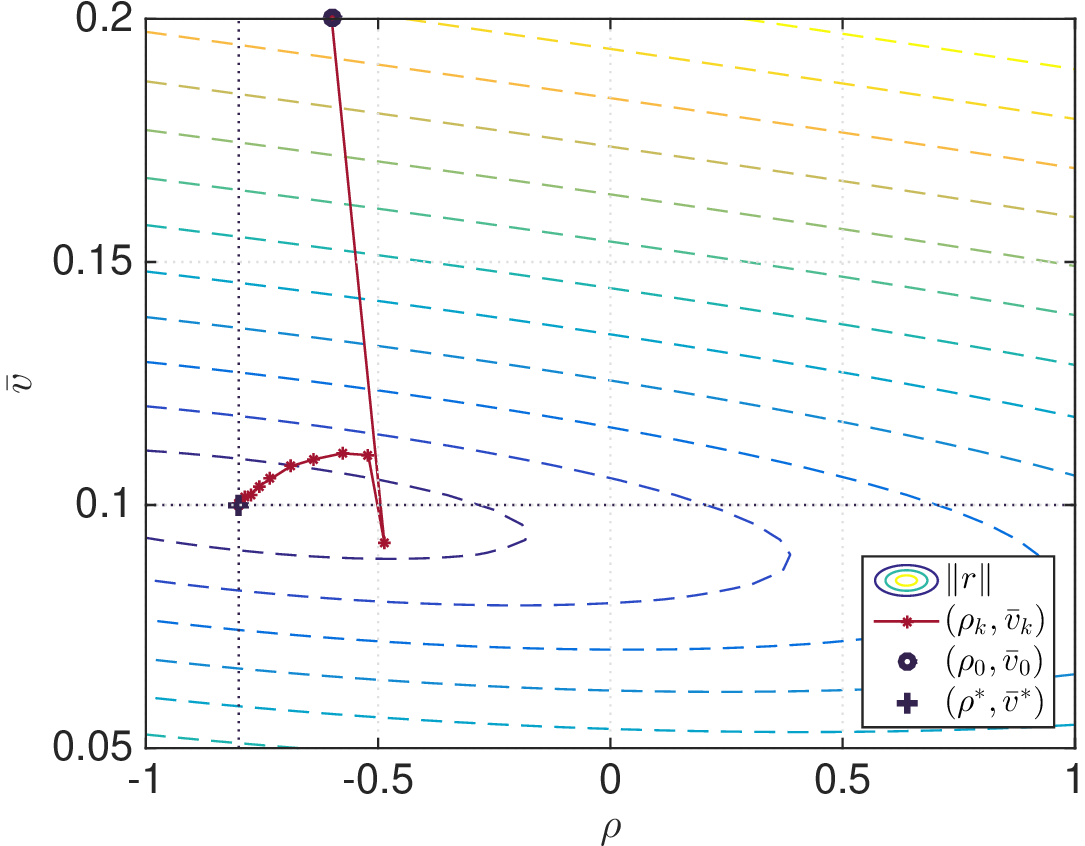}\label{fig:cont_rhovinf}}\\
\caption{Contours of $\|\bm{r}\|$ and iteration path for $\left(\theta_i, \theta_j\right)$.}\label{fig:contour}
\end{figure}
\begin{figure}[H]
\centering
\ContinuedFloat
\subfloat[]{\includegraphics[width=0.48\textwidth]{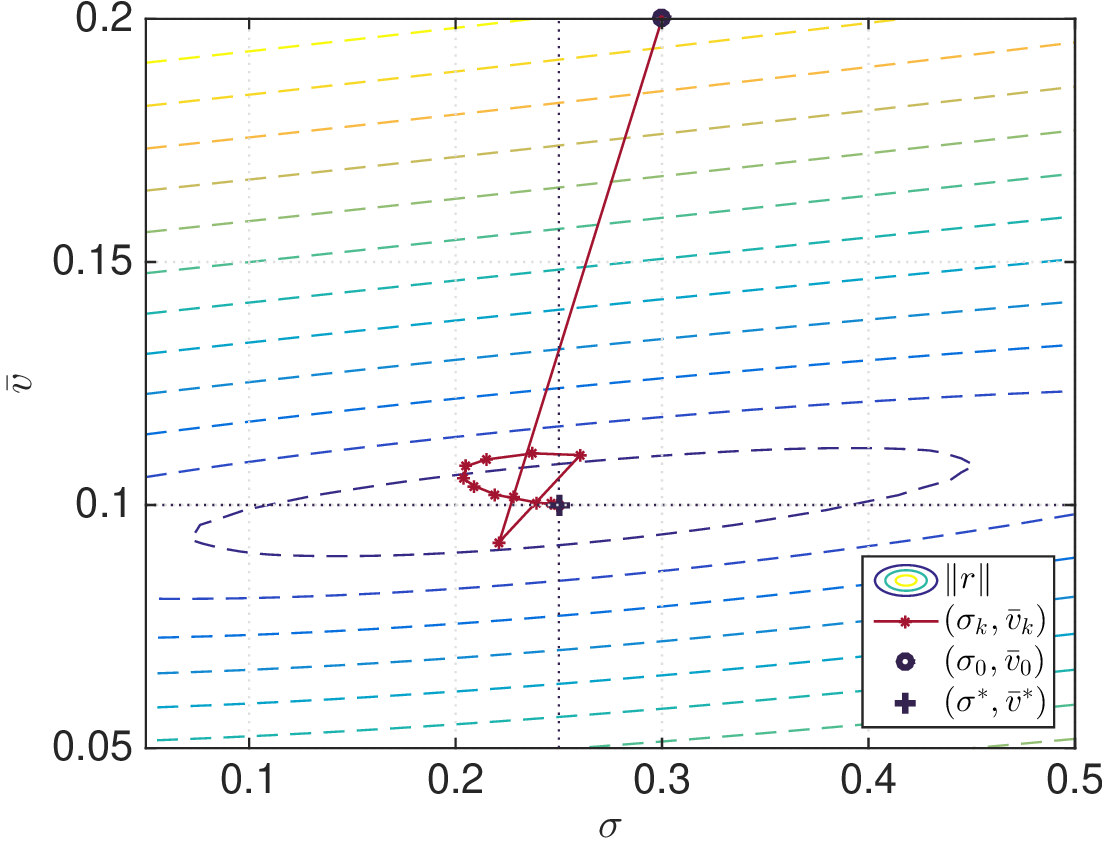}\label{fig:cont_vovvinf}}\hfill
\subfloat[]{\includegraphics[width=0.48\textwidth]{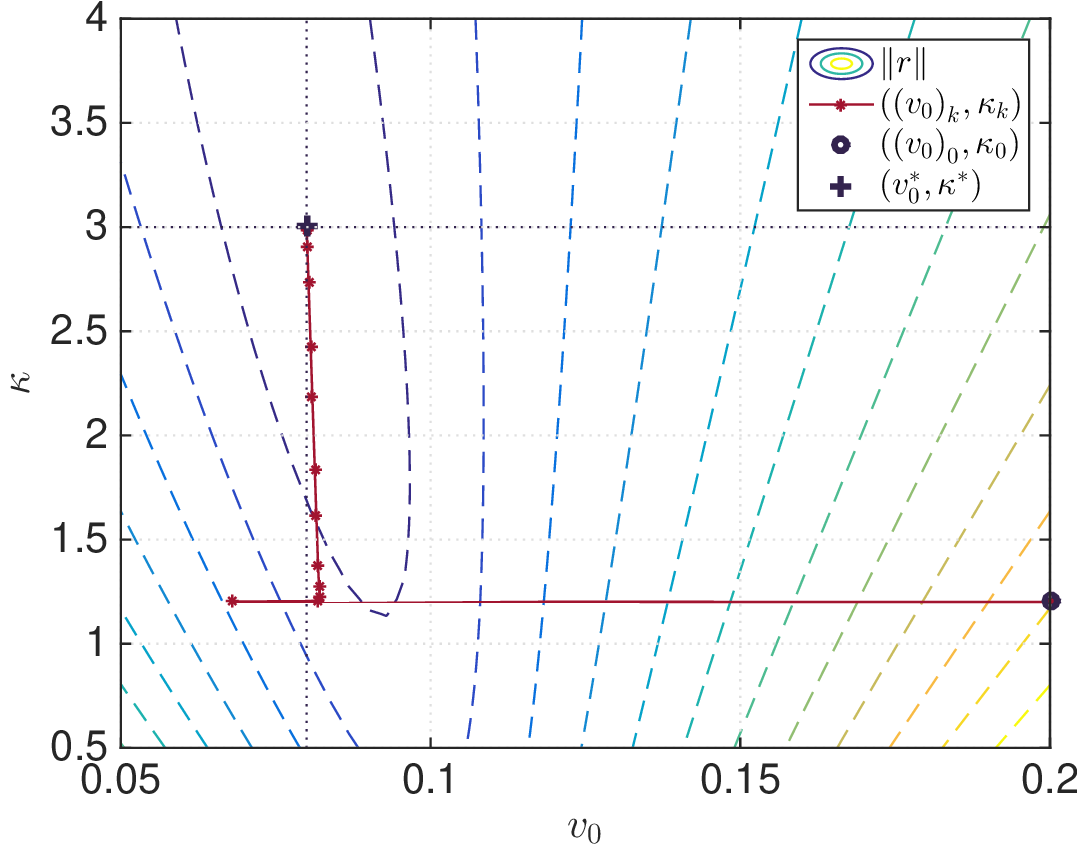}\label{fig:cont_v0kap}}\\
\subfloat[]{\includegraphics[width=0.48\textwidth]{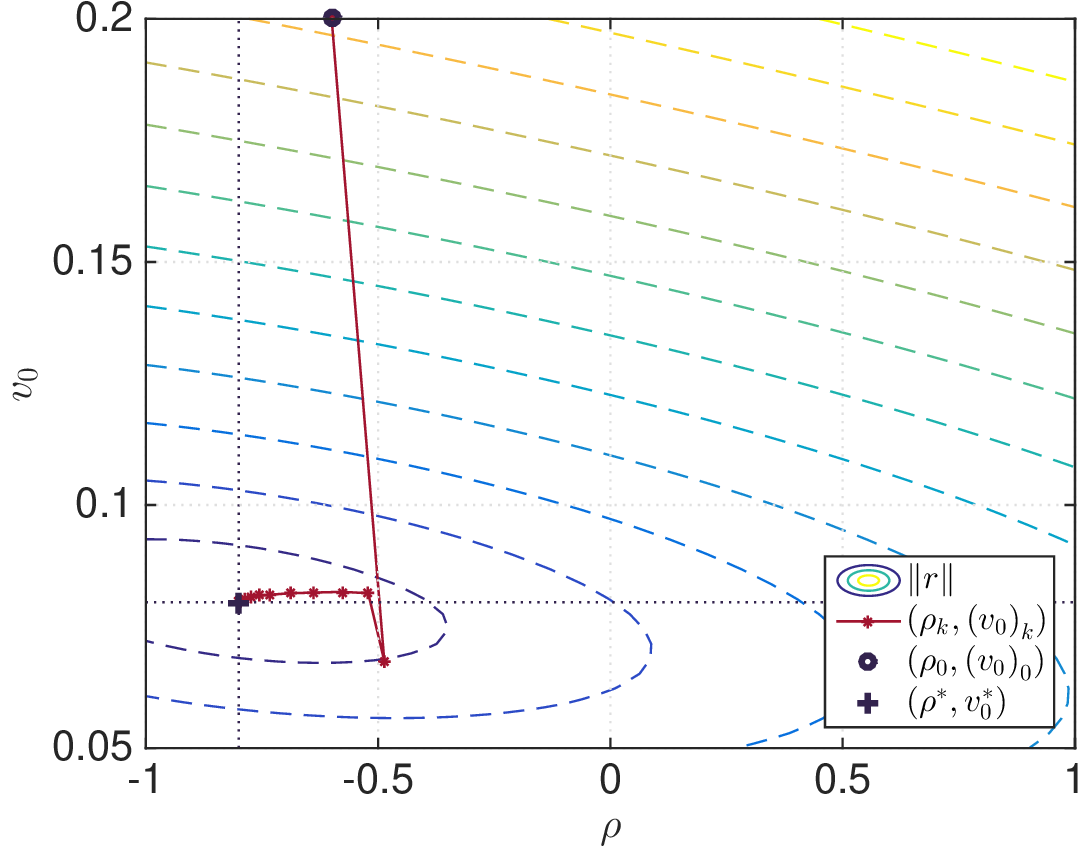}\label{fig:cont_rhov0}}\hfill
\subfloat[]{\includegraphics[width=0.48\textwidth]{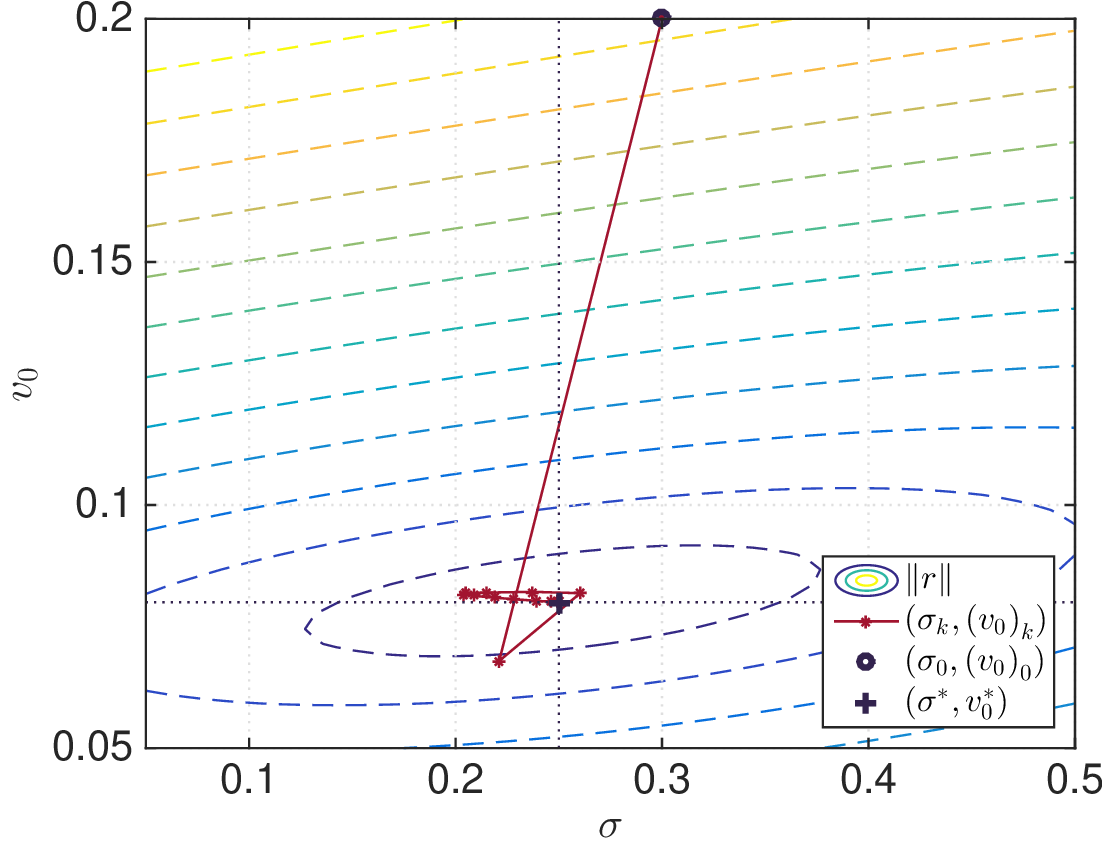}\label{fig:cont_vovv0}}\\
\subfloat[]{\includegraphics[width=0.48\textwidth]{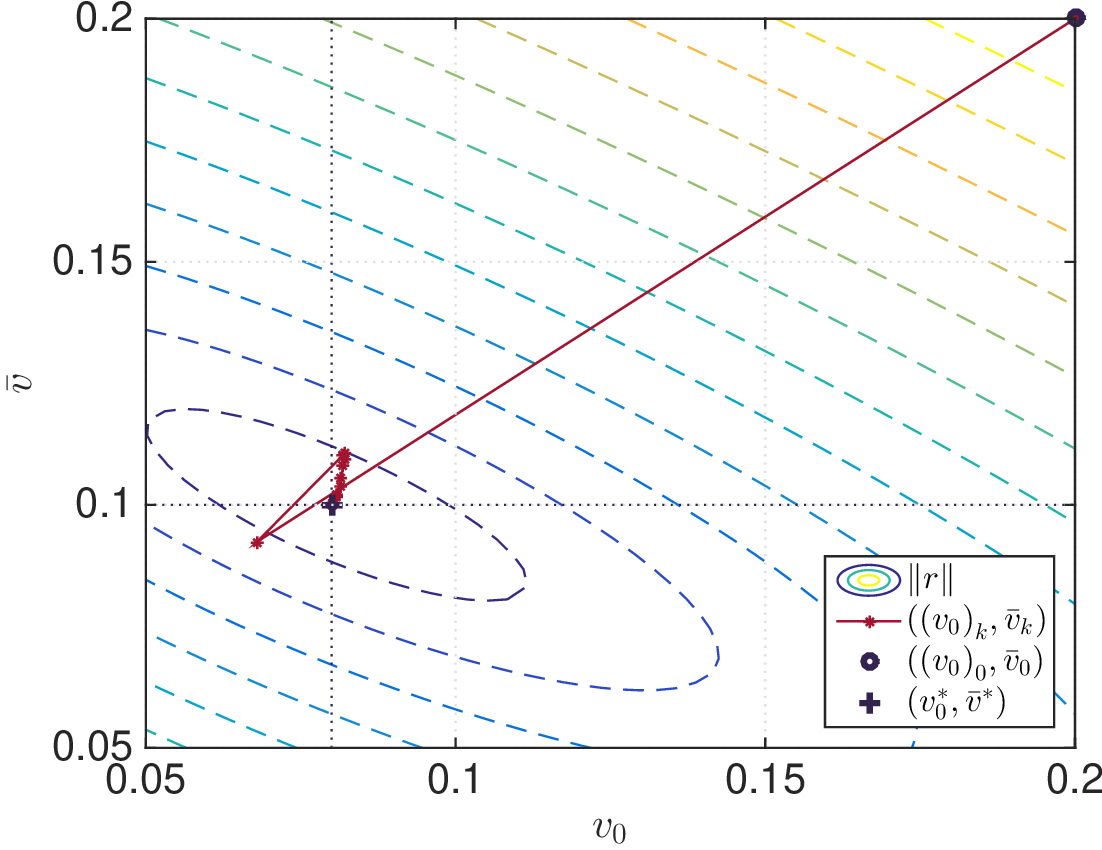}\label{fig:cont_v0vinf}}\hfill
\subfloat[]{\includegraphics[width=0.48\textwidth]{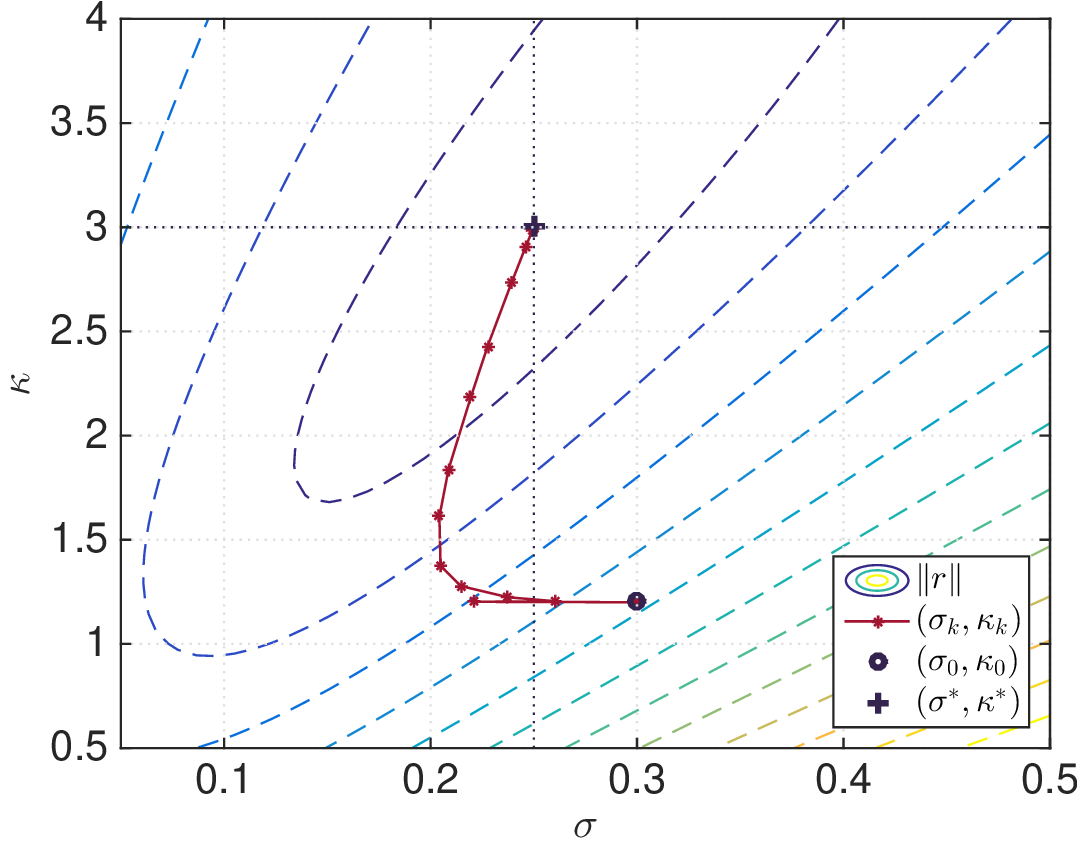}\label{fig:cont_vovkap}}\\
\caption{(cont.) Contours of $\|\bm{r}\|$ and iteration path for $\left(\theta_i, \theta_j\right)$.}%\label{fig:contour}
\end{figure}
\begin{figure}[H]
\centering
\ContinuedFloat
\subfloat[]{\includegraphics[width=0.48\textwidth]{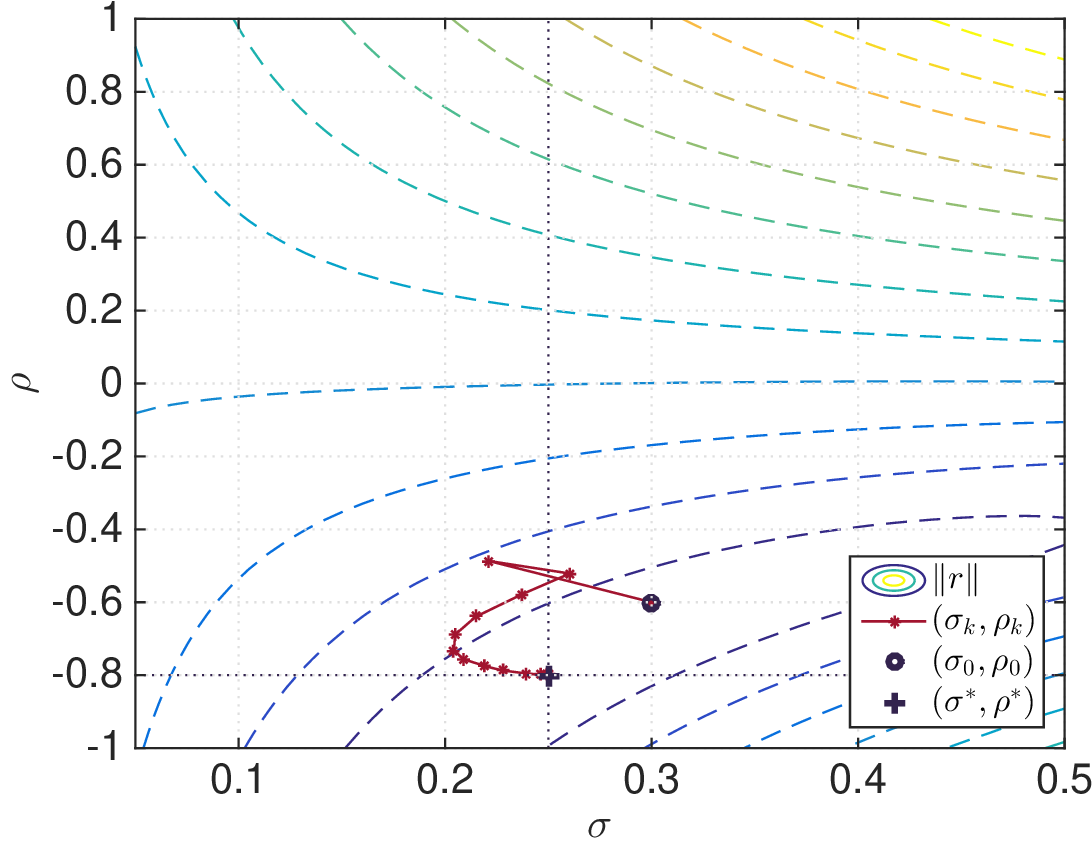}\label{fig:cont_vovrho}}\hfill
\subfloat[]{\includegraphics[width=0.48\textwidth]{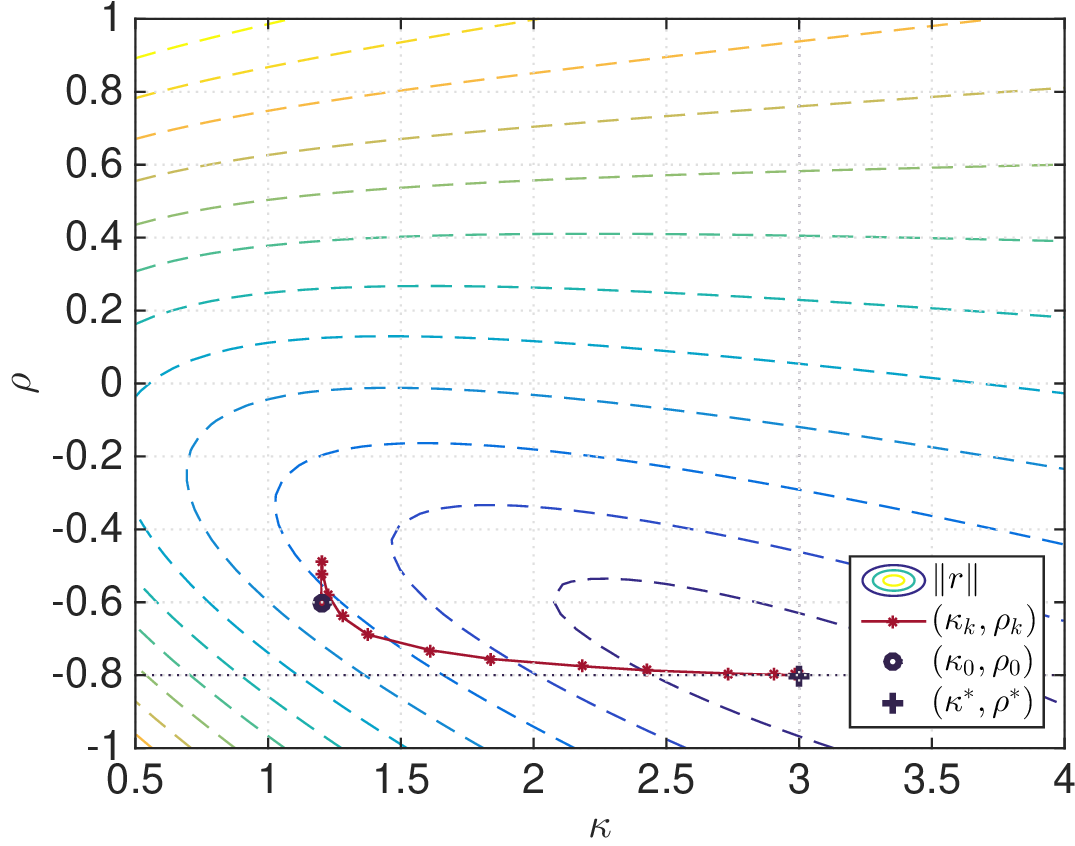}\label{fig:cont_kaprho}}
\caption{(cont.) Contours of $\|\bm{r}\|$ and iteration path for $\left(\theta_i, \theta_j\right)$.}
\end{figure}

The Gauss-Newton approximation of the Hessian matrix at the optimal solution is given in Table~\ref{tab:Hessian}.
\begin{table}[h]
 \caption{The Hessian matrix $\bm{\nabla\nabla}^\intercal f(\bm{\theta}^*)$.}\label{tab:Hessian}
 \begin{center}\footnotesize
 \renewcommand{\arraystretch}{1.3}
 \begin{tabular}{l | r  r  r  r  r}\hline\hline
 & $\partial\kappa$ & $\partial\bar{v}$ & $\partial\sigma$ & $\partial\rho$ & $\partial v_0$\\ \hline
$\partial\kappa$ & $5.26\times 10^{-5}$  &	 &     &      &   \\
$\partial\bar{v}$ & $9.65\times 10^{-3}$   & $2.26\times 10^{+1}$	 &   	 &      &     \\
$\partial\sigma$ & $-5.49\times 10^{-4}$ & $-7.66\times10^{-2}$  & $7.46\times 10^{-3}$	&      &       \\
$\partial\rho$     & $1.61\times 10^{-4}$	&  $2.00\times 10^{-2}$   &	$-2.34\times 10^{-3}$	& $7.56\times 10^{-4}$ & \\
$\partial v_0$     & $5.28\times 10^{-3}$	&  $1.18\times 10^{+1}$       &	$-3.53\times 10^{-2}$	& $8.40\times 10^{-3}$	 & $9.69\times 10^{-1}$\\
\hline\hline
\end{tabular}
 \end{center}
\end{table}

The Hessian matrix is ill-conditioned with a condition number of $3.978 \times 10^{6}$.
The elements $\partial^2 f(\bm{\theta}^*)/\partial \kappa^2$ and $\partial^2 f(\bm{\theta}^*)/\partial \rho^2$ are of a much smaller order than the others.
This suggests that the objective function, when around the optimum, is \textit{less sensitive} to changes along $\kappa$ and $\rho$.
In other words, the objective function is more stretched along these two axes as can be verified looking at the contours, for example in
Figs.~\ref{fig:cont_kapvinf} and \ref{fig:cont_rhovinf}.
The ratio between $\partial^2 f(\bm{\theta}^*)/\partial \kappa^2$ and $\partial^2 f(\bm{\theta}^*)/\partial \bar{v}^2$ is of order $10^{-6}$, which indicates a great disparity in sensitivity:
changing $1$ unit of $\bar{v}$ is comparable to changing $10^{6}$ units of $\kappa$.
On the other hand, this explains the so-called local minima reported in previous research.
When one starts from a different initial point and stops the iteration with a high tolerance, it is possible that the iterate lands somewhere in the region where $\kappa$ and $\rho$ are very different.
There are two possible approaches that one can seek to deal with this: the first is to scale the parameters to a similar order and search on a better-scaled objective function; the second is to decrease the tolerance level for the optimisation process, meaning to approach the very bottom of this objective function.

In Table~\ref{tab:solverCompare}, we present the performance of the LM method with analytical gradient (LMA), the LM method with numerical gradient (LMN), and a feasibility perturbed sequential quadratic programming method (FPSQP) \cite{gerlich2012parameter} adopted in UniCredit bank. As the concrete implementation of FPSQP is owned by the bank, we only extract their test results.
%The number of iteration is given; note that the starting point and optimal parameter set for LMA is the same as that for LMN, but different from that for FPSQP, thus the number of iteration for FPSQP is not comparable and not put in the table.
The computational cost can be compared through the number of evaluations of the pricing function \eqref{eq:HestonCall} per iteration, expressed as a multiple of the number $n$ of options to be calibrated.
 \begin{table}[h]
 \caption{Performance comparison between solvers.}
 \begin{center}\footnotesize
 \renewcommand{\arraystretch}{1.3}
 \begin{tabular}{l | r  r r  }\hline\hline
    & LMA & LMN & FPSQP\\ \hline
    Stopping criterion &  $\|\bm{r}(\bm{\theta}_k)\|\leq10^{-10}$&  $\|\bm{r}(\bm{\theta}_k)\|\leq10^{-10}$& $\|\bm{\Delta}\bm{\theta}_k \|\leq 10^{-6}$\\
    Iterations & $13$ & $22$ &-\\
    Price evaluations per iteration & $1.08n$ & $1.70n$ & $6.00n$\\
\hline\hline
 \end{tabular}\label{tab:solverCompare}
 \end{center}
 \end{table}
LMA requires about $n$ pricing function evaluation per step.
LMN requires more for the gradient approximation, but the difference is not large since LMN uses a rank-one update for the subsequent Jacobian matrices.
FPSQP requires about 5.5 times more than that of LMA and achieves only a lower accuracy for the stopping criterion for the gradient.

We tested our method also on a few realistic model parameterisations. In Table~\ref{tab:typicalcases}, we present three test cases that are representative respectively for long-dated FX options, long-dated interest rate options and equity options \cite{andersen2007efficient}. They are believed to be prevalent and challenging for the simulation of Heston model \cite{glasserman2011gamma}.
\begin{table}[h]
 \caption{Test cases with realistic Heston model parameters. Case I: long-dated FX options. Case II: long-dated interest rate options. Case III: equity options.}
 \begin{center}\footnotesize
 \renewcommand{\arraystretch}{1.3}
 \begin{tabular}{l | r r r }\hline\hline
    & Case I & Case II &  Case III \\ \hline
 $\kappa^*$ & 0.50     & 0.30    & 1.00   \\
 $\bar{v}^*$  & 0.04   & 0.04   & 0.09 \\
 $\sigma^*$  & 1.00     &0.90      & 1.00 \\
 $\rho^*$      & -0.90    & -0.50    & -0.30     \\
 $v_0^*$      & 0.04    &0.04    & 0.09 \\ \hline\hline
 \end{tabular}\label{tab:typicalcases}
 \end{center}
 \end{table}
Each component of the initial guess is an independent uniformly distributed random number in the $\pm10\%$ range of the corresponding optimum.
This choice is due to the fact that practitioners usually choose the initial guess as the last available estimation which is expected to be close to the solution if the calibration is frequent enough and the market does not change drastically.
%The frequent intraday calibration of the Heston model calibrate the Heston model frequently within the same trading day, and
%in practice the initial guess is often chosen as the last available estimate and may not vary too much from the current value if the time interval between estimations is too short for the market to change severely.
We test each case with 100 initial guesses.
Our previous test range in Table~\ref{tab:parameter_initialpoint} has covered these cases too, but here we would like to focus on the performance of our method when applied to these typical examples and thus justify its computational efficiency and robustness for practical application.
The information about the convergence as an average of the 100 initial guesses is given in Table~\ref{tab:typicalcases_result}. For the practical cases with initial guesses in the vicinity, it takes less than or around one second to obtain the optimal solution. % For Case I, the calibration process is often stopped by stagnating iterates, and the residual is not reduced to a similar level as for the other two cases.

 \begin{table}[ht]
 \caption{Calibration results for three typical realistic cases, reporting an average on 100 initial guesses for each of them.}
 \begin{center}\footnotesize
 \renewcommand{\arraystretch}{1.3}
 \begin{tabular}{ l l | r r r } \hline\hline
   & & Case I           & Case II & Case III \\ \hline
  & $|\kappa^\dag - \kappa^*|$       & $2.87\times 10^{-2}$      & $1.35\times 10^{-3}$   & $1.20\times 10^{-3}$  \\
\multirow{1}{*}{Absolute}& $|\bar{v}^\dag-\bar{v}^*|$       & $4.80\times 10^{-3}$   &  $4.52\times 10^{-5}$  & $2.11\times 10^{-5}$  \\
\multirow{1}{*}{deviation}  & $|\sigma^\dag-\sigma^*|$         & $5.29\times 10^{-2}$      & $7.48\times 10^{-4}$    &  $3.94\times 10^{-4}$ \\
\multirow{1}{*}{from $\bm{\theta}^*$}& $|\rho^\dag-\rho^*|$             & $3.65\times 10^{-2}$     & $1.69\times 10^{-5}$  &  $1.46\times 10^{-5}$  \\
& $|v_0^\dag-v_0^*|$              & $2.14\times 10^{-3}$     & $1.46\times 10^{-5}$   &  $1.07\times 10^{-5}$ \\ \hline
 & $\|\bm{r}_0\|$ & $2.70\times 10^{-4}$  & $4.51\times 10^{-5}$ &$1.02\times 10^{-4}$ \\
\multirow{1}{*}{Error} & $\|\bm{r}^\dag\|$  & $1.12\times 10^{-4}$ &$9.24\times 10^{-11}$ & $3.33\times 10^{-11}$ \\
\multirow{1}{*}{measure}  &$\|{\bm{J}^\dag}\bm{e}\|_\infty$ & $1.77\times 10^{-1}$ & $4.63\times 10^{-6}$& $4.15\times 10^{-6}$\\
 & $\|\Delta\bm{\theta}^\dag\|$ & $6.88\times 10^{-21}$ & $1.63\times 10^{-8}$ & $5.10\times 10^{-5}$\\ \hline
& CPU time & 0.40 & 1.11 & 0.15\\
\multirow{1}{*}{Computational}  &  LM iterations             & 16.83     & 51.52   & 6.86\\
\multirow{1}{*}{cost}  & Price evaluations       & 23.38     & 52.60   & 7.86\\
 &  Gradient evaluations & 16.83     & 51.52   & 6.86\\
 &  Linear systems solved & 23.38   & 51.52   & 6.86 \\ \hline\hline
\end{tabular}\label{tab:typicalcases_result}
\end{center}
\end{table}

\section{Conclusion}\label{sec:conclusion}
We proposed a new representation of the Heston characteristic function which is continuous and easily derivable.
We derived the analytical form of the gradient of the Heston option pricing function with respect to the model parameters.
The result can be applied in any gradient-based algorithm.
An algorithm for a full and fast calibration of the Heston model is given.
The LM method succeeds in finding the global optimal parameter set within a reasonable number of iterations.
The method is validated by randomly generated parameterisations as well as three typical cases of Heston model parameterisations for long-dated FX options, long-dated interest rate options and equity options.
The resulting parameters can replicate the volatility surface with an $l_2$-norm error of $10^{-10}$ and an $l_1$-norm error around $10^{-7}$.
The cheap computational cost and the stable performance for different initial guesses make the proposed method suitable for the purpose of high-frequency trading. Several numerical issues are discussed.
We also present the final Hessian matrix and contours of the objective function. We point out that either a rescaling of the parameters or a low tolerance level is needed to find the global optimum.

\section*{Acknowledgement}
We thank Gianluca Fusai and Giuseppe Di Poto for useful comments. The support of the Economic and Social Research Council (ESRC) in funding the
Systemic Risk Centre is gratefully acknowledged (grant number ES/K002309/1).

\bibliography{CalibHeston}
\bibliographystyle{model5-names}

\end{document}